\documentclass[a4paper,reqno]{amsart}

\usepackage{a4wide}
\usepackage[foot]{amsaddr}
\usepackage{etoolbox}

\usepackage[utf8]{inputenc}
\usepackage[T1]{fontenc}
\usepackage{amsmath}
\usepackage{amssymb,empheq,mathtools,dsfont}
\usepackage{physics}
\usepackage{stmaryrd}
\usepackage[mathscr]{eucal}
\usepackage{color}
\usepackage{nicefrac}
\usepackage{thm-restate}
\usepackage{calrsfs}
\usepackage{hyperref}
\DeclareMathAlphabet{\pazocal}{OMS}{zplm}{m}{n}

\usepackage{scalerel,stackengine}
\stackMath
\newcommand\reallywidehat[1]{\savestack{\tmpbox}{\stretchto{\scaleto{\scalerel*[\widthof{\ensuremath{#1}}]{\kern.1pt\mathchar"0362\kern.1pt}{\rule{0ex}{\textheight}}}{\textheight}}{2.4ex}}\stackon[-6.9pt]{#1}{\tmpbox}}

\newcommand{\Bc}{{\mathcal B}}
\newcommand{\Cpc}{{\pazocal C}}

\newcommand{\Sc}{{\mathcal S}}

\newcommand{\C}{\mathbb{C}}
\newcommand{\F}{\mathbb{F}}

\newcommand{\R}{\mathbb{R}}
\newcommand{\Z}{\mathbb{Z}}

\newcommand{\CC}{\ensuremath{\mathscr{C}}}

\renewcommand{\vec}[1]{\mathbf{#1}}

\newcommand{\cv}{\vec{c}}

\renewcommand{\ev}{\vec{e}}

\newcommand{\xv}{\vec{x}}
\newcommand{\yv}{\vec{y}}

\newcommand{\cov}[1]{{\left| #1\right|}}

\newcommand{\ogv}{\omega_{\textup{GV}}}

\newcommand{\dmin}{d_{\textup{min}}}

\DeclareMathOperator*{\argmin}{argmin}

\makeatletter
\newcommand*{\transp}{{\mathpalette\@transpose{}}}
\newcommand*{\@transpose}[2]{\raisebox{\depth}{$\m@th#1\intercal$}}
\makeatother

\newcommand*{\eqdef}{\stackrel{\text{def}}{=}}

\newtheorem{theorem}{Theorem}[section]
\newtheorem{corollary}[theorem]{Corollary}
\newtheorem{remark}[theorem]{Remark}
\newtheorem{definition}[theorem]{Definition}
\newtheorem{proposition}[theorem]{Proposition}
\newtheorem{lemma}[theorem]{Lemma}
\newtheorem{notation}[theorem]{Notation}
\newtheorem{fact}[theorem]{Fact}

\newcommand{\vol}[1]{V_{n}\left(#1\right)}

\newcommand\dual[1]{#1^{*}}

\usepackage{todonotes}

\newcommand{\CKL}{C_{\textup{KL}}}

\newcommand{\Nb}[2]{N_{#1}( #2 )}
\newcommand{\Neq}[2]{\Nb{#1}{#2}}
\newcommand{\Neqs}[1]{N_{#1}}

\newcommand{\crand}[2]{\Cpc_{#1,#2}}

\newcommand{\uniff}{u_{\textup{full}}}
\newcommand{\unifq}{u}
\newcommand{\unifc}{u_{\CC}}

\newcommand{\unifs}[1]{u_{#1}}

\newcommand{\fber}[1]{f_{\textup{ber},#1}}

\newcommand{\fberTrunc}[1]{f_{\textup{truncBer},#1}}

\newcommand{\eball}{\mathcal B}

\newcommand{\gunif}[1]{u_{#1 \eball}}

\author{Thomas Debris-Alazard$^{1}$} \email{thomas.debris@inria.fr}  
\author{L\'eo Ducas$^{2,3}$} \email{L.Ducas@cwi.nl} 
\author{Nicolas Resch$^{4}$} \email{n.a.resch@uva.nl} 
\author{Jean-Pierre Tillich$^{1}$} \email{jean-pierre.tillich@inria.fr}
\address{$^{1}$ Inria}
\address{$^{2}$ CWI, Amsterdam, The Netherlands}
\address{$^{3}$ Mathematical Institute, Leiden University}
\address{$^{4}$ Informatics' Institute, University of Amsterdam}
\thanks{The work of TDA and JPT was funded by the French Agence Nationale de la
	Recherche through ANR JCJC COLA (ANR-21-CE39-0011) for TDA and ANR CBCRYPT  (ANR-17-CE39-0007) for JPT. Part of this work was done while NR was affiliated with the CWI and partially supported by ERC H2020 grant No.74079 \mbox{(ALGSTRONGCRYPTO)}. LD is supported by an ERC starting Grant 947821 (ARTICULATE)} 

\title{Smoothing Codes and Lattices: \\ Systematic Study and New Bounds}

\begin{document}
	
	\maketitle	 
	
	\begin{abstract}
		In this article we revisit smoothing bounds in parallel between lattices {\em and} codes. Initially introduced by Micciancio and Regev, these bounds were instantiated with Gaussian distributions and were crucial for arguing the security of many lattice-based cryptosystems. Unencumbered by direct application concerns, we provide a systematic study of how these bounds are obtained for both lattices {\em and} codes, transferring techniques between both areas. We also consider multiple choices of spherically symmetric noise distribution.

		We found that the best strategy for a worst-case bound combines Parseval's Identity, the Cauchy-Schwarz inequality, and the second linear programming bound, and this holds for both codes and lattices and all noise distributions at hand. For an average-case analysis, the linear programming bound can be replaced by a tight average count.

		This alone gives optimal results for spherically uniform noise over random codes and random lattices. This also improves previous Gaussian smoothing bound for worst-case lattices, but surprisingly this provides even better results with uniform ball noise than for Gaussian (or Bernoulli noise for codes).

		This counter-intuitive situation can be resolved by adequate decomposition and truncation of Gaussian and Bernoulli distributions into a superposition of uniform noise, giving further improvement for those cases, and putting them on par with the uniform cases.
	\end{abstract}

\section{Introduction}

\subsection{Smoothing bounds.} In either a code or a lattice, smoothing refers to fact that, as an error distribution grows wider and wider, the associated syndrome distribution tends towards a uniform distribution. In other words, the error distribution, reduced modulo the code or the lattice, becomes essentially flat. This phenomenon is pivotal in arguing security of cryptosystems~\cite{MR07,GPV08,DST19}. In information theoretic literature, it is also sometimes referred to as flatness~\cite{LLBS14}.
Informally, by a ``smoothing bound'' we are referring to a result which lower bounds the amount of noise which needs to be added so that the smoothed distribution ``looks'' flat. 

To be more concrete, by a ``flat distribution'', we are referring to a uniform distribution over the ambient space modulo the group of interest. For a (linear) code $\CC \subseteq \F_2^n$, this quotient space is $\F_2^n/\CC$; for a lattice $\Lambda \subseteq \R^n$, it is $\R^n/\Lambda$. We then consider some ``noise'' vector $\vec e$ distributed over the ambient space $\F_2^n$ ({\em resp.} $\R^n$), and attempt to prove that $\vec e \mod \CC$ ({\em resp.} $\vec e \mod \Lambda$) is ``close'' to the uniform distribution over the quotient space $\F_2^n/\CC$ ({\em resp.} $\R^n/\Lambda$). To quantify ``closeness'' between distributions, we will use the standard choice of \emph{statistical distance}.

An important question to be addressed is the choice of distribution for the noise vector $\vec e$. 
In lattice-based cryptography (where such smoothing bounds originated~\cite{MR07}), the literature ubiquitously uses Gaussian distributions for errors, and smoothness is guaranteed for an error growing as the inverse of the minimum distance of the dual lattice. The original chain~\cite{MR07} of argument goes as follows: 
\begin{itemize}
	\item Apply the Poisson summation formula (PSF);
	\item Bound variations via the triangle inequality (TI) over all non-zero dual lattice points;
	\item Bound the absolute sum above via the Banaszczyk tail bound~\cite{B93} for discrete Gaussian (BT).
\end{itemize}
An intermediate quantity called the smoothing parameter introduced by~\cite{MR07} before the last step is also often used in the lattice-based cryptographic literature. Each bounding step is potentially non-tight, and indeed more recent works have replaced the last step by the following~\cite{ADRS15}:
\begin{itemize}
	\item Bound the number of lattice points in balls of a given radius via the Linear Programming bound~\cite{L79} (LP) and ``sum over all radii'' (with care).
\end{itemize}
With this LP strategy, it is in principle possible to also compute a smoothing bound for spherically symmetric distributions of errors other than the Gaussian; however, we are not aware of prior work doing this explicitly. A very natural choice would be uniform distributions over Euclidean balls.

For codes, there are also two natural distributions of errors: Bernoulli noise, {\em i.e.} flip each bit independently with some probability $p$ ({\em a.k.a.} the binary symmetric $\mathrm{BSC}_p$ channel), and a uniform noise over a Hamming sphere of a fixed radius. The latter is typically preferred for the design of concrete and practical cryptosystems~\cite{M78,A11,MTSB13,DST19}, while the former appears more convenient in theoretical works \footnote{A third choice of distribution, described as a discrete-time random walk, also made an appearance for a complexity theoretic result~\cite{BLVW19}. The expert reader may note that the Bernoulli distribution can also be treated as a continuous-time random walk, and both can be analysed via the heat kernel formalism ~\cite[Chap. 10]{C97}.}. Cryptographic interest for code smoothing has recently arisen~\cite{BLVW19,YZ20}, but results are so far limited to codes with extreme parameters and specific ``balancedness'' constraints. However we note that the question is not entirely new in the coding literature (see for instance ~\cite{K07}). In particular, an understanding of the smoothing properties of Bernoulli noise is intimately connected to the \emph{undetected error probability} of a code transmitted through the $\mathrm{BSC}_p$. 

In this light, it is interesting to revisit and systematize our understanding of smoothing bounds, unencumbered by direct application concerns. We find it enlightening to do this exploration in parallel between codes and lattices, transferring techniques back and forth between both areas whenever possible. 

Furthermore, we keep our arguments agnostic to the specific choice of error distribution, allowing us to apply them with different error distributions and compare the results. To compare different (symmetric) distributions, we advocate parametrizing them by the expected weight/norm of a vector. That is, we quantify the magnitude of a noise vector $\vec e$ by $t = \mathbb E(|\vec e|)$ (where $|\cdot|$ denotes either the Hamming weight or the Euclidean norm of the vector). Our smoothing bounds will depend on this parameter, and we consider a smoothing bound to be more effective if for the smoothed distribution to be close to uniform we require a smaller lower-bound on $t$. 

\subsection{Contributions.}

In this work, we collect the techniques that have been used for smoothing, both in the code and lattice contexts. We view individual steps as modular components of arguments, and consider all permissible combinations of steps, thereby determining the most effective arguments. In the following, we outline our systematization efforts, describing the various proof frameworks that we tried before settling on the most effective argument. 

\paragraph{\bf Code smoothing bounds.}

Given the relative dearth of results concerning code smoothing, it seems natural to start by adapting the first argument (PSF+TI+BT) to codes following the proof techniques of~\cite{B93,MR07}. And indeed, the whole strategy translates flawlessly, with only one caveat: it leads to a very poor result, barely better than the trivial bound. Namely, smoothness is established only for Bernoulli errors with parameter very close to $p=1/2$.

The adaptation of Banaszczyk tail bound~\cite{B93} to codes (together with replacing the Gaussian by a Bernoulli distribution) is rather na\"ive, and it is therefore not very surprising that it leads to a disappointing result. Instead, we can also follow the improved strategy for lattices from~\cite{ADRS15}, and resort to linear programming bounds for codes~\cite{B65,MRRW77,ABL01}.
Briefly, by an LP bound we are referring to a result that bounds the number of codewords ({\em resp.} lattice vectors) of a certain weight ({\em resp.} norm) in terms of the dual distance ({\em resp.} shortest dual vector) of the code ({\em resp.} lattice). In both cases, the results are obtained by considering a certain LP relaxation of the combinatorial quantities one wishes to bound, hence the name. Even more, the bounds for codes and lattices are obtained via essentially the same arguments \cite{MRRW77,DL98,CE03}. We therefore find it natural to apply LP bounds in our effort to develop proof techniques which apply to both code- and lattice-smoothing. 

The strategy (PSF+TI+LP) turns out to give a significantly better result, but it nevertheless still appears to be far from optimal. We believe that the application of the triangle inequality in the second step to bound the sum of Fourier coefficients given by the Poisson summation formula leads to the unsatisfactory bound. Indeed, a common heuristic when dealing with sums of Fourier coefficients is that, unless there is a good reason otherwise, the sum should have magnitude roughly the square-root of the order of the group (as is the case for random signs): the triangle inequality is far too crude to notice this.

Instead, we turn to another common upper-bound on a sum, namely, the Cauchy-Schwarz (CS) inequality. It is natural to subsequently apply Parseval's Identity (PI). It turns out that this strategy yields very promising results, upon which we now elucidate. The upper-bound is described in terms of the \emph{weight distribution} of a code, {\em i.e.} the number of codewords of weight $w$ for each $w=1,\dots,n$. Unfortunately, it is quite difficult to understand the weight distribution of arbitrary codes, and the bounds that we do have are quite technical. 

\paragraph{\bf Random codes.} For this reason, we first apply our proof template to \emph{random codes}, as it is quite simple to compute the (expected) weight distribution of a random code. 
Quite satisfyingly,  the simple two steps arguments (PI+CS) already yields \emph{optimal} results for this case, but when the error is sampled uniformly at random from a sphere! That is, we can show that the support size of the error distribution matches the obvious lower bound that applies to \emph{any} distribution that successfully smooths a code: namely, for a code $\CC$ the support size must be at least $\sharp(\F_2^n/\CC)$. Using coding-theoretic terminology, the weight of the error vector that we need to smooth is given by the ubiquitous Gilbert-Varshamov bound
\[
	\omega_{\textup{GV}}(R) =  h^{-1}(1-R)
\] 
which characterizes the trade-off between a random code's rate $R$ and its minimum distance. Here, $h^{-1}$ is the inverse of the binary entropy function.

Moreover, as the argument is versatile enough to apply to essentially all spherical error distributions, 
we also tried applying it to the Bernoulli distribution, and the random walk distribution of~\cite{BLVW19}. Comparing them, we were rather surprised that our argument provided better bounds for the uniform distribution over a Hamming sphere than the other two distributions for the same average Hamming weight.

However, while the (PI+CS) sequence of arguments is more effective when the noise is sampled uniformly on the sphere, we can exploit the fact that the Hamming weight of a Bernoulli-distributed vector is tightly concentrated to recover the same smoothing bound for this distribution. In more detail, we use a ``truncated'' argument. First, we decompose the Bernoulli distribution into a convex combination of uniform sphere distributions. But, by Chernoff's bound, a Bernoulli distribution is concentrated on vectors whose weight lies in a width $\varepsilon n$ interval around its expected weight. Therefore, outside of this interval, the contribution of the Bernoulli on the statistical distance is negligible. Then apply the (PI+CS) sequence of arguments to each constituent distribution close to the expected weight. In this way, we are able to demonstrate that Bernoulli distributions also optimally smooth random codes.

\paragraph{\bf Arbitrary codes.} Next, we turn our attention to smoothing worst-case codes. Motivated by our success in smoothing random codes, we again follow the (PI+CS) sequence of arguments and combine this with LP bounds to derive smoothing bounds when the dual distance of the code is sufficiently large. Again, the sequence of arguments is most effective when the error is distributed uniformly over the sphere, with one caveat: we are also required to assume that the dual code is \emph{balanced} in the sense that it also does not contain any vectors of too large weight. While this assumption has appeared in other works~\cite{BLVW19,YZ20}, we find it somewhat unsatisfactory.

Fortunately, this condition is not required if the error is sampled according to the Bernoulli distribution. But then we run into the same issue that we had earlier with random codes: the (PI+CS) argument, followed by LP bounds, natively yields a lesser result when instantiated with Bernoulli noise. Fortunately, we have already seen how to resolve this issue: we pass to the truncated Bernoulli distribution and decompose it into uniform sphere distributions. This yields a best-of-both-worlds result: we obtain the strongest smoothing bound we can in terms of the noise magnitude, while requiring the weakest assumption on the code.

\paragraph{\bf And back to lattices.} Having now uncovered this better strategy for codes, we can return to lattices and apply our new proof template. Indeed, as we outline in Section~\ref{subsec:fourier-analysis}, the (PI+CS) sequence of arguments can be applied in a very broad context; see, in particular, Corollary~\ref{coro:FB}.

\paragraph{\bf Random lattices.} First, just as we set our expectations for code-smoothing by first studying the random case, we analogously start here by considering random lattices. However, defining a random lattice is a non-trivial task. We actually consider two distributions. The first, which is based on the deep Minkowski-Hlwaka-Siegel (MHS) Theorem, we only abstractly describe. Thanks to the MHS Theorem, we can very easily compute the (expected value) of our upper-bound. 

For the MHS distribution of lattices, we consider two natural error distributions: the Gaussian distribution (which is used ubiquitously in the literature), as well as the uniform distribution over the Euclidean ball. And again, perhaps surprisingly (although less so now thanks to our experience with the code case), we obtain a better result with the uniform distribution over the Euclidean ball. And moreover, the Euclidean ball result is \emph{optimal} in the same sense that we had for codes: the support volume of the error distribution is exactly equal to the covolume of the lattice \footnote{That is, for a lattice $\Lambda$, the volume of the torus $\R^n/\Lambda$. We will denote this quantity by 
$\cov{\Lambda}$ from now on. }. We view the value $w$ such that the volume of the $n$-ball of radius $w$ is equal to the covolume of a lattice (which is half the quantity that appears in Minkowski bound) as being the lattice-theoretic analogue of the Gilbert-Varshamov quantity: 
\[
	w_{\textup{M}/2} \eqdef \frac{\sqrt[n]{\cov{\Lambda}
\; \Gamma(n/2+1)}}{\sqrt{\pi}} \ .
\]

However, as Gaussian vectors satisfy many pleasing properties that are often exploited in lattice-theoretic literature, we would like to obtain the same smoothing bound for this error distribution. Fortunately, our experience with codes also tells us how to recover the result for Gaussian noise from the Euclidean ball noise smoothing bound: we decompose the Gaussian distribution appropriately into a convex combination of Euclidean ball distributions. Together with a basic tail bound, we recover the same smoothing bound for Gaussian noise that we had for the uniform ball noise.

We also study random $q$-ary lattices, which are more concretely defined: following the traditional lattice-theoretic terminology, they are obtained by applying Construction A to a random code. This does lead to a slight increase in the technicality of the argument -- in particular, we need to apply a certain ``summing over annuli'' trick -- but the computations are still relatively elementary. Again, we find that the argument naturally works better when the errors are distributed uniformly over a ball, but we can still transfer the bound to the Gaussian noise.

Interestingly, the same optimal bound has been recovered in a concurrent work \cite[Theorem 1.]{LLB22} for Gaussian distributions. Their arguments are quite unlike ours: \cite{LLB22} uses the Kullback–Leibler divergence in combination with other information-theoretic arguments. However, contrary to our bounds obtained via the (PI + CS) sequence of arguments, \cite[Theorem 1]{LLB22} only holds for random $q$-ary lattices. 

\paragraph{\bf Arbitrary lattices.} Next, we address the challenge of smoothing arbitrary lattices. And again, we follow the (PI+CS) sequence of arguments, and subsequently use the Kabatiansky and Levenshtein bound~\cite{KL78} to obtain a smoothing bound in terms of the minimum distance of the dual lattice. The Kabatiansky and Levenshtein bound is the lattice-analogue of the second LP bound from coding theory. We can directly apply the arguments with both of our error distributions of interest, and again, the uniform ball distribution wins. But the decomposition and tail-bound trick again applies to yield the same result for the Gaussian distribution that we had for the uniform ball distribution.

\paragraph{\bf Comparison.} We summarize how our work improves on the state of the art in Table~\ref{tab:lattice_smoothing_bound} for lattices, and in Table~\ref{tab:code_smoothing_bound} and Figure~\ref{figure:figureCompSmoothingRandCode} for codes. For this discussion, we let $U(\R^n/\Lambda)$ ({\em resp.} $U(\F_2^n/\CC$)) denote the uniform distribution over $\R^n/\Lambda$ ({\em resp.} $\F_2^n/\CC)$, and let $\Delta$ denote the statistical distance. 

In the case of lattices (Table~\ref{tab:lattice_smoothing_bound}), we fix the smoothing bound target to exponentially small, that is we state the minimal value of $F>0$ such that the bound over the statistical distance implies $\Delta(\vec{e} \bmod \Lambda, U(\mathbb R^n / \Lambda)) \leq 2^{-\Omega(n)}$ when the error follows the prescribed distribution and of an average Euclidean length of $\mathbb E(|\vec{e}|_{2}) = F \; n / \lambda_1^*(\Lambda)$.\footnote{In fact, the values in this table guarantee exponentially small statistical distance from the uniform distribution.}

\begin{table}
	\begin{tabular}{l|l|c|l}
		Distribution & Proof strategy & smoothing factor $F$ & General statement \\ \hline
		Gaussian & PSF+TI+BT & $1/(2\pi) \approx 0.15915$		& 
Lemma 3.2 ~\cite{MR07} \\
		Gaussian & PSF+TI+LP & $\CKL/(2\pi \sqrt{e}) \approx  0.12746$		& 
Lemma 6.1 ~\cite{ADRS15} \\
		Gaussian & PI+CS+LP & $\CKL/(2\pi \sqrt{2e}) \approx 0.09013$		& Theorem~\ref{theo:bSDEgauss} (this work) \\ \hline
		Unif. Euclidean ball & PI+CS+LP & $\CKL/(2\pi e) \approx 0.07731$		& Theorem~\ref{theo:bSDEuc} (this work) \\ \hline
		Gaussian & via Unif. + Trunc. & $\CKL/(2\pi e) \approx 0.07731$		& Theorem~\ref{theo:bSDEgauss_better} (this work)
	\end{tabular}
	\caption{Comparison of smoothing bounds for various proof strategies and error distributions. The smoothing constant $F$ is the smallest constant $C$ such that the bounds proves exponential smoothness when the average norm (over $n$, the length of the ambient space) of an error is at least $C$ times the inverse of the minimal distance of the dual lattice. Here $\CKL \approx 2^{0.401}$ denotes the constant that is involved in the Kabatiansky and Levenshtein bound \cite{KL78}.\label{tab:lattice_smoothing_bound}}
\end{table}

In the case of codes we also fix the smoothing bound target to negligible,\footnote{Again, it is the same if we insist the statistical distance to uniform is exponentially small.} but we compare two cases: smoothing bounds for random codes (in average) and for a fixed code (worst case). In Figure \ref{figure:figureCompSmoothingRandCode} we compare the minimal value $F>0$ such that $\mathbb{E}_{\CC}\left(\Delta(\vec{e} \bmod \CC, U(\mathbb \F_{2}^n / \CC))\right) \leq 2^{-\Omega(n)}$ when the error $\vec{e}$ follows the prescribed distribution and with an expectation that is taken over codes of rate $R$. In Table \ref{tab:lattice_smoothing_bound} we make the same comparison but to reach $\Delta(\vec{e} \bmod \CC, U(\mathbb \F_{2}^n / \CC)) \leq 2^{-\Omega(n)}$ for a fixed code $\CC$ such that the minimum distance of its dual $\dual{\CC}$ is known.

	\begin{center}
	\begin{figure}
		\includegraphics[height=6cm]{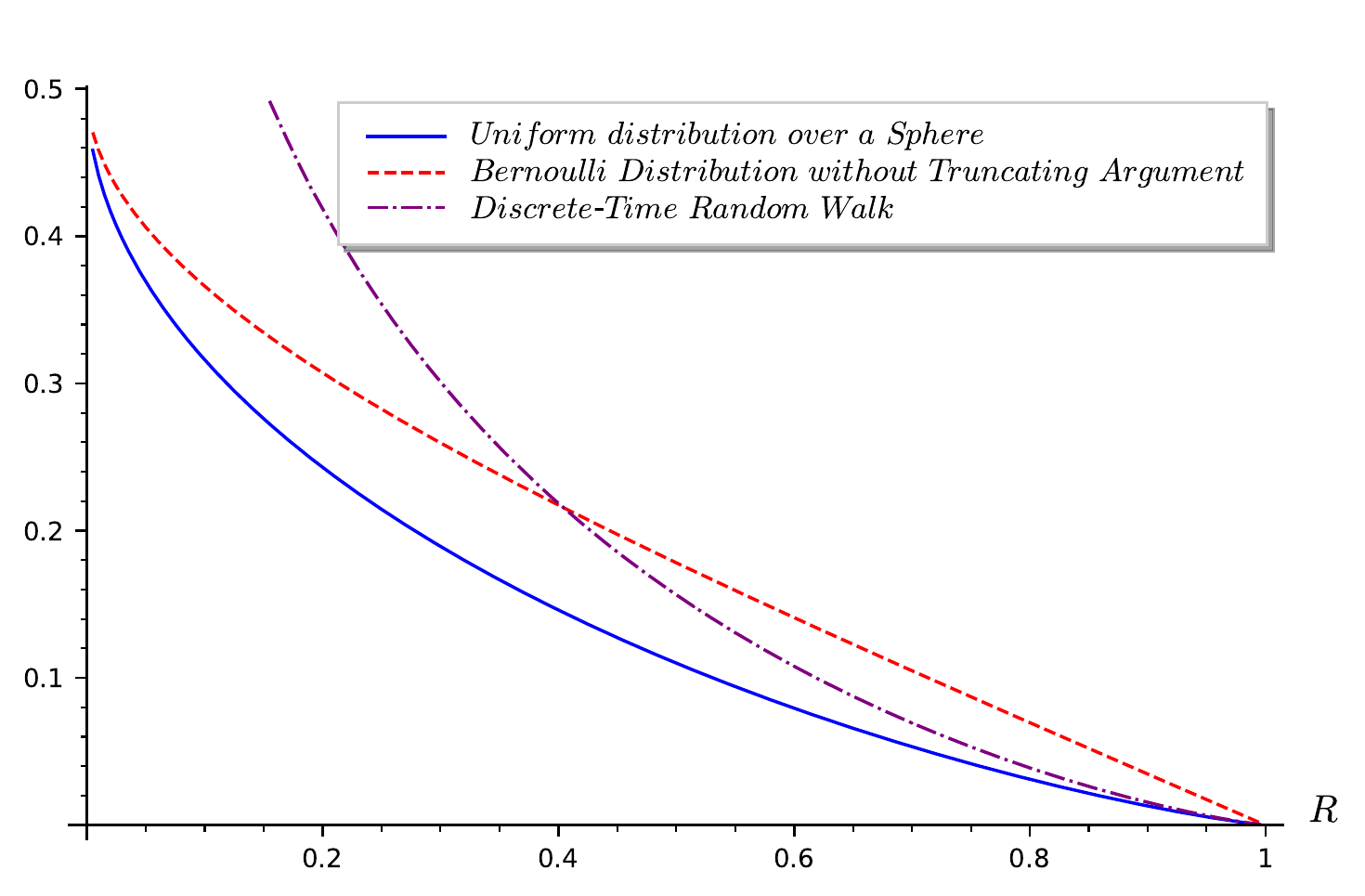}
		\caption{Comparison of smoothing constants for {\em random} codes as a function of their rate $R$ for various error distributions. The smoothing constant is the smallest constant $C$ such that the bounds proves exponential smoothness when the average Hamming weight of an error is at least $C n$.} \label{figure:figureCompSmoothingRandCode}
	\end{figure}
\end{center}

\begin{table}
	\begin{tabular}{l|c|c|l}
		Distribution & smoothing factor $F$ & Balanced-code & General statement \\ \hline
		Bernoulli & $\approx 0.24$		& NO & Eq.~\eqref{eq:BwithBer}, Prop.~\ref{propo:ABL}, \ref{propo:2LPB} \\ \hline
		Discrete Rand. Walk & $\approx 0.27$ & YES & Theorem~\ref{theo:smoothingBoundsUnifRW} \\ \hline
		Unif. Hamming sphere & $\approx 0.17$	& YES & Theorem~\ref{theo:smoothingBoundsUnifRW} \\ \hline
		Bernoulli + Trunc. &	$\approx 0.17$	& NO & Theorem~\ref{theo:finalUBSD}  \\
	\end{tabular}
	\caption{Comparison of smoothing bounds for a code $\CC$ of length $n$ such that its dual $\dual{\CC}$ has minimum distance $0.11n$ (which is the typical case for a code of rate $1/2$) for various error distributions. The smoothing constant $F$ is the smallest constant $C$ such that the bounds proves exponential smoothness when the average Hamming weight of an error is at least $Cn$. Furthermore the balanced-code hypothesis means that we suppose there are no dual codewords $\dual{\vec{c}}\in\dual{\CC}$ of Hamming weight larger than $(1-0.11)n$.\label{tab:code_smoothing_bound}} 
\end{table}

 	\section{Preliminaries: Notations and Fourier Analysis over Locally Compact Abelian Group}\label{sec:prelim}

	\subsection{General Notation.} The notation $x \eqdef y$ means that $x$ is defined as being equal to $y$. Given a set $\mathcal{S}$, its indicator function will be denoted $1_{\mathcal{S}}$. For a finite set $\mathcal{S}$, we will denote by $\sharp\mathcal{S}$ its cardinality. Vectors will be written with bold letters (such as $\vec{x}$). Furthermore, we denotes by $\llbracket a,b \rrbracket$ the set of integers $\{a,a+1,\dots,b\}$.

	The statistical distance between two discrete probability distributions $f$ and $g$ over a same space $\mathcal{S}$ is defined as:
	$$
	\Delta(f,g) \eqdef \frac{1}{2} \sum_{x \in \mathcal{S}} |f(x) - g(x)|.
	$$
	Similarly, for two continuous probability density functions $f$ and $g$ over a same measure space  $\mathcal{E}$, the statistical distance is defined as 
	$$
	\Delta(f,g) \eqdef \frac{1}{2} \int_{\mathcal{E}} |f - g|. 
	$$

	\subsection{Codes and Lattices} We give here some basic definitions and notation about linear codes and lattices.

	{\bf \noindent Linear codes.} In the whole paper, we will deal exclusively with binary linear codes, namely subspaces of $\F_2^n$ 
	for some positive integer $n$.
The space $\F_{2}^{n}$ will be embedded with the Hamming weight $|\cdot|$, namely
	$$
	\forall \vec{x}\in\F_{2}^{n}, \quad |\vec{x}| \eqdef \sharp \left\{ i \in \llbracket 1,n \rrbracket \mbox{ : } x_{i} \neq 0 \right\}. 
	$$
	We will denote by $\mathcal{S}_{w}$ the sphere with center $\mathbf{0}$ and radius $w$; its size is given by $\binom{n}{w}$ and we have $\frac{1}{n}\log_2 \binom{n}{w} = h(w/n) + o(1)$
	where $h$ denotes the binary-entropy, namely $h(x) \eqdef -x \log_2 (x) - (1-x) \log_{2}(1-x)$.

	An $\lbrack n ,k \rbrack$-code $\CC$ is defined as a dimension $k$ subspace of $\F_{2}^{n}$. The rate of $\CC$ is $\frac{k}{n}$. Its minimal distance is given by 
	\begin{align*}
	\dmin(\CC) &\eqdef \min \left\{ |\cv - \cv'| \text{ :  }\cv,\cv'\in \CC \text { and } \cv \neq \cv' \right\} \\
	&= \min \left\{ |\vec{c}| \mbox{ : } \vec{c}\in \CC \mbox{ and } \vec{c} \neq \mathbf{0} \right\}.
	\end{align*}
The number of codewords of $\CC$ of weight $t$ will be denoted by $\Neq{t}{\CC}$,
namely
	$$
	\Neq{t}{\CC} \eqdef \sharp \left\{ \vec{c} \in \CC \mbox{ and } |\vec{c}| = t  \right\}. 
	$$
	The dual of a code $\CC$ is defined as $\dual{\CC} \eqdef \left\{ \dual{\vec{c}} \in \F_{2}^{n} \mbox{ : } \forall \vec{c}\in\CC, \mbox{ } \vec{c}\cdot\dual{\vec{c}} = 0 \right\}$
	where $\cdot$ denotes the standard inner product on $\F_{2}^{n}$.

	{\bf \noindent Lattices.} We will consider lattices of $\mathbb{R}^{n}$ which is embedded with the Euclidean norm $|\cdot|_{2}$, namely
	$$
	\forall \vec{x}\in\mathbb{R}^{n}, \quad |\vec{x}|_{2} \eqdef \sqrt{\sum_{i=1}^{n} x_{i}^{2}}. 
	$$
	We will denote by $\mathcal{B}_{w}$ the ball with center $\mathbf{0}$ and radius $w$; its volume is given by
$$
	\vol{w} \eqdef \frac{\pi^{n/2}w^{n}}{\Gamma(n/2+1)}.
	$$
An $n$-dimension lattice $\Lambda$ is defined as a discrete subgroup of $\mathbb{R}^{n}$. The covolume $|\Lambda| \eqdef \textup{vol}\left( \mathbb{R}^{n}/\Lambda\right)$ of $\Lambda$ is the volume of any fundamental parallelotope. The minimal distance of $\Lambda$ is given by
	$
	\lambda_{1}(\Lambda) \eqdef \min \left\{ |\vec{x}|_{2} \mbox{ : } \vec{x}\in \Lambda \mbox{ and } \vec{x} \neq \mathbf{0} \right\}. 
	$
	The number of lattice points of $\Lambda$ of weight $\leq t$ will be denoted by $\Nb{\leq t}{\Lambda}$, namely
	$$
	\Nb{\leq t}{\Lambda} \eqdef \sharp \left\{ \vec{x} \in \Lambda \mbox{ : } |\vec{x}|_{2} \leq  t  \right\}. 
	$$

	\subsection{Fourier Analysis} \label{subsec:fourier-analysis}

	We give here a brief introduction to Fourier analysis over arbitrary locally compact Abelian groups. Our general treatment will allow us to apply directly some basic results in a code and lattice context, obviating the need in each case to introduce essentially the same definitions and to provide the same proofs.

	Corollary \ref{coro:FB} at the end of this subsection is the starting point of our smoothing bounds: all of our results are obtained by using different facts 
to bound the right hand side of the inequality. 
	\newline

	{\bf \noindent Groups and Their Duals.} In what follows $G$ will denote a locally compact Abelian group. Such a group admits a Haar measure $\mu$.  For instance $G = \mathbb{R}$ with $\mu$ the Lebesgue measure $\lambda$, or $G = \F_{2}^{n}$ with $\mu$ the counting measure $\sharp$.

	The dual group $\widehat{G}$ is given by the continuous group homomorphisms $\chi$ from
	$G$ into the multiplicative group of complex numbers of absolute value $1$, and it is again
	a locally compact Abelian group.
In Figure \ref{fig:groups} we give groups, their duals as well as their associated Haar measures that will be considered in this work.

	\begin{figure}
		\renewcommand{\arraystretch}{1.4}.
		\begin{tabular}{|c|c||c|c|}
			\hline
			$G$ & $\mu$ & $\widehat{G}$ & $\mu$ \\
			\hline\hline
			$\mathbb{F}_{2}^{n}$ & $\frac{1}{2^{n}} \; \sharp$ &  & \\
			$\mathbb{F}_{2}^{n}/\CC$ & $ \frac{\sharp\CC}{2^{n}}\; \sharp$ & $\widehat{\mathbb{F}_{2}^{n}/\CC} \simeq \dual{\CC}$ & $\sharp$ \\
			$\CC$ & $\frac{1}{\sharp\CC}\;\sharp$ &  & \\
			\hline
			$\mathbb{R}^{n}$ & $\lambda$ &  &  \\
			$\mathbb{R}^{n}/\Lambda$ & $\frac{1}{|\Lambda|}\; \lambda$ & $\widehat{\mathbb{R}^{n}/\Lambda} \simeq \dual{\Lambda}$ & $\sharp$ \\
			$\Lambda$ & $\sharp\; |\Lambda|$ &  &  \\
			\hline
		\end{tabular}
		\caption{Some groups $G$, their duals $\widehat{G}$ and their associated Haar measures. Here $\lambda$ denotes the Lebesgue measure and $\sharp$ the counting measure.}
		\label{fig:groups}
	\end{figure}

	It is important to note that if $H \subseteq G$ is a closed subgroup, then $G/H$ and $H$ are also locally compact groups. Furthermore, $G/H$ has a dual group that satisfies the following isomorphism 
	$$
	\widehat{G/H} \simeq H^{\perp} \eqdef \left\{ \chi \in \widehat{G} \mbox{ : } \forall h \in H, \mbox{ } \chi(h) = 1 \right\}.
	$$

	{\bf \noindent Norms and Fourier Transforms.} For any $p \in [1, \infty[$,
	$L_{p}(G)$ will denote the space of measurable functions $f : G \rightarrow \mathbb{C}$ (up to functions which agree almost everywhere) with finite norm $\| f \|_{p}$ which is defined as
	$$
	\| f \|_{p} \eqdef \sqrt[p]{\int_{G} |f|^{p} d\mu}.
	$$
	The Fourier transform of $f \in L_{1}(G)$ is defined as
	$$
	\widehat{f} : \chi \in \widehat{G} \longmapsto \int_{G}f\overline{\chi}d\mu.
	$$
	We omitted here the dependence on $G$. It will be clear from the context.

	\begin{theorem}[Parseval's Identity]\label{theo:Parseval} Let $f\in L_{1}(G) \cap L_{2}(G)$, then with appropriate normalization of the Haar measure
		$$
		\| f \|_{2} = \| \widehat{f} \|_{2}.
		$$
	\end{theorem}

	{\bf \noindent  Poisson Formula.} Given $H \subseteq G$ and any function $f : G \rightarrow \mathbb{C}$, its restriction over $H$ is defined as $f_{|H}: h \in H \mapsto f(h)\in \mathbb{C}$. We define its periodization as follows.
	
	\begin{definition}[Periodization]\label{def:perio}
		Let $H$ be a closed subgroup of $G$ and $f\in L^{1}(G)$. We define the $H$-periodization of $f$ as
		$$
		f^{|H} : (g+H) \in G/H \longmapsto \int_{H} f(g+h)d\mu_{H}(h) \in \mathbb{C}
		$$
		where $\mu_{H}$ denotes any choice of the Haar measure for $H$. 
	\end{definition}

	There always exists a Haar measure $\mu_{G/H}$ such that for any continuous function with compact support $f : G \rightarrow \mathbb{C}$ the quotient integral formula holds
	\begin{equation}\label{eq:quoIntFormula} 
	\int_{G/H} \left( \int_{H} f(g+h)d\mu_{H}(h)   \right) d\mu_{G/H}(g + H) = \int_{G} f(g)d\mu(g). 
	\end{equation}

	\begin{theorem}[Poisson Formula] Let $H\subseteq G$ be a closed subgroup and $f\in L^{1}(G)$, then with appropriate normalization of the Haar measures,
		$$
		\widehat{\left( f^{|H} \right)} = \left(\widehat{f}\right)_{|\widehat{G/H}}. 
		$$
	\end{theorem}

	The following corollary is a simple consequence of the Cauchy-Schwarz inequality, Parseval identity and the Poisson formula. Our results on smoothing bounds are all based on this corollary.

	\begin{corollary}\label{coro:FB} Let $H$ be a closed subgroup of $G$. Let $a : x \in G/H \mapsto 1$ and $f \in L^{1}(G)$ such that $\int_{G} f d\mu = \mu_{G/H}(G/H)$. Then with appropriate normalization of the Haar measure,\footnote{We choose the Haar measures $\mu_{G}$, $\mu_{H}$,$\mu_{G/H}$ and $\widehat{\mu_{G/H}}$ for which both the Poisson formula and Parseval's Identity hold.}
		$$
		\| a - f^{|H} \|_{1} \leq \sqrt{\mu_{G/H}(G/H)} \; \sqrt{\int_{\widehat{G/H} \backslash \{\chi_{\mathbf{0}}\}} |\widehat{f}|^{2} \; d\mu_{\widehat{G/H}}}
		$$ 
	where $\chi_{\mathbf 0}$ denotes the identity element of $\widehat{G/H}$. 
	\end{corollary}
	
	\begin{proof}
		We have 
		\begin{align}
			\|a-f^{|H}\|_{1} &= \int_{\widehat{G/H}}|a-f^{|H}|d\mu_{G/H} \nonumber \\
&\leq \sqrt{\mu_{G/H}(G/H)} \; \|a-f\|_2  \quad (\mbox{By Cauchy-Schwarz})\nonumber \\
&= \sqrt{\mu_{G/H}(G/H)} \; \|\widehat{a}-\widehat{f}\|_2 \quad (\mbox{By Parseval}) \nonumber\\
&= \sqrt{\mu_{G/H}(G/H)} \; \sqrt{\int_{\widehat{G/H}\setminus\{\chi_{\mathbf{0}}\}}|\widehat{f^{|H}}|^2d\mu_{\widehat{G/H}}} \label{eq:fourier-values}\\
			&= \sqrt{\mu_{G/H}(G/H)} \; \sqrt{\int_{\widehat{G/H}\setminus\{\chi_{\mathbf{0}}\}}|\widehat{f}|^2d\mu_{\widehat{G/H}}} 2 \quad (\mbox{By Poisson})\nonumber
		\end{align}
		where in Equation \eqref{eq:fourier-values} we used the following equalities:
\begin{align*}
				\widehat{f^{|H}}(\chi_{\mathbf 0}) &= \int_{G/H}f^{|H}\overline{\chi_{\mathbf 0}}\; d\mu_{G/H} \\
				&=  \int_{G/H}\left( \int_{H}f(g+h)d\mu_{H}(h)\right)d\mu_{G/H}(g+H) \\
				& = \int_{G} f \quad (\mbox{By Equation \eqref{eq:quoIntFormula}})\\ 
				& =  \mu_{G/H}(G/H) \quad (\mbox{By assumption on $f$})
			\end{align*}
			and
			$$
			\widehat{a}(\chi_{\mathbf{0}}) = \int_{G/H}u\overline{\chi_{\mathbf 0}}d\mu_{G/H}  = \mu_{G/H}(G/H) \quad \mbox{and} \quad \forall  \chi\in \widehat{G/H}\setminus\{\chi_{\mathbf{0}}\}, \mbox{ } \widehat{a}(\chi) =\int_{G/H}\overline{\chi}d\mu_{G/H} = 0 .
			$$
			which concludes the proof. 
\end{proof}

	In this work we will choose $G = \mathbb{R}^{n}$ and $H = \Lambda$ or $G = \mathbb{F}_{2}^{n}$ and $H = \CC$. Haar measures associated to $G, G/H$ and $\widehat{G/H}$ for which the corollary holds are given in Figure \ref{fig:groups}. Furthermore, we will use Fourier transforms over $\widehat{G}$ and $\widehat{G/H}$. We describe in Figure \ref{table:FT} these dual groups that we will consider. 
	
	\begin{figure}[htb]
	\begin{center} 
			\renewcommand{\arraystretch}{1.8}
	\begin{tabular}{|c|c|}
		\hline
			$\R^n$ & $\F_2^n$ \\ \hline
$\widehat{\mathbb{R}^{n}/\Lambda} =  \left\{ \chi_{\vec{x}} \mbox{ : } \vec{x}\in\dual{\Lambda} \right\}$ & $\widehat{\mathbb{F}_{2}^{n}/\CC} =  \left\{ \chi_{\vec{x}}\mbox{ : } \vec{x}\in\dual{\CC} \right\}$ \\  \hline
	$\widehat{f}(\xv) = \int_{\R^n} f(\yv) e^{2i\pi \xv \cdot \yv} d\yv$ & $ \widehat{f}(\xv) = \frac{1}{2^n} \sum_{\yv \in \F_2^n} f(\yv) (-1)^{\xv \cdot \yv} $ \\		
		\hline
\end{tabular}
	
	\end{center} 
	\caption{Dual groups and Fourier transforms that we will consider. We identify $\widehat{f}(\chi_{\vec{x}})$ with $\widehat{f}(\vec{x})$. \label{table:FT}}
\end{figure}

 	\section{Smoothing Bounds: Code Case}\label{sec:SBCode}

Given a binary linear code $\CC$ of length $n$,
the aim of a smoothing bound is to quantify at which condition on the noise $\cv+\ev$ is statistically close to 
the uniform distribution over $\F_2^n$  when 
$\cv$ is uniformly drawn from $\CC$ and $\ev$ sampled according to some noise distribution $f$.
Equivalently, we want to understand when $\left(\vec{e} \mod \CC\right) \in \F_{2}^{n}/\CC$ is close to the uniform distribution.
We will focus on the case where the distribution of $\ev$ is radial, meaning that it only depends on the Hamming weight 
of $\ev$.

\begin{notation}
	We will use throughout this section the following notation.
\begin{itemize}
\item The uniform probability distribution over the quotient space $\F_{2}^{n}/\CC$ will frequently recur and for this reason we just denote it by $\unifq$. The uniform distribution over the whole space $\F_2^n$ is denoted by $\uniff$ and the uniform distribution over the codewords of $\CC$ is denoted by $\unifc$. 
\item We also use the uniform distribution over
 the sphere 
 $\Sc_w$ 
which we denote by $\unifs{w}$.
\item For two probability distributions $f$ and $g$ over $\F_2^n$ we denote by $f \star g$ the convolution over $\F_2^n$: 
$f \star g(\xv) = \sum_{\yv \in \F_2^n} f(\xv-\yv)g(\yv)$.
\end{itemize}
\end{notation}

It will be more convenient to work in the quotient space and for this we use the following proposition. 
\begin{proposition} Let $f$ be a probability distribution over $\F_2^n$ and $\CC$ be an $\lbrack n,k \rbrack$-code. We have
	\begin{equation*}
	\Delta(\uniff,\unifc \star f) = \Delta(u,f^{\CC}),  \quad \mbox{where } f^{\CC}(\xv) \eqdef 2^{k} \; f^{|\CC}(\xv) = \sum_{\cv \in \CC} f(\xv-\cv).
	\end{equation*}
\end{proposition}
\begin{proof}Let $\vec{c}$ and $\vec{e}$ be distributed according to $u_{\CC}$ and $f$. We have the following computation:
\begin{align}
	\Delta(\uniff,\unifc \star f)& = \frac{1}{2}\; \sum_{\vec{x}\in\F_{2}^{n}} \left| \frac{1}{2^{n}} - \mathbb{P}_{\unifc,f}\left( \vec{c} + \vec{e} = \vec{x} \right) \right| \nonumber \\
	&=  \frac{1}{2}\;\sum_{\vec{x}\in\F_{2}^{n}} \left| \frac{1}{2^{n}} - \sum_{\vec{c}_{0}\in\CC} \mathbb{P}_{f}(\vec{c}+\vec{e} = \vec{x}\mid \vec{c} = \vec{c}_{0})\;\frac{1}{2^{k}} \right| \nonumber \\
	&=  \frac{1}{2}\;\sum_{\vec{x}\in\F_{2}^{n}} \left| \frac{1}{2^{n}} - \frac{1}{2^{k}}\;\sum_{\vec{c}_{0}\in\CC} f(\vec{x} - \vec{c}_{0}) \right|  \nonumber \\
	&=  \frac{1}{2}\; \sum_{\vec{x}\in\F_{2}^{n}/\CC} \left| \frac{1}{2^{n-k}} - \sum_{\vec{c}_{0}\in \CC} f(\vec{x}-\vec{c}_{0}) \right| \label{eq:modC} \\
	& =  \frac{1}{2}\; \sum_{\vec{x}\in\F_{2}^{n}/\CC} \left| \frac{1}{2^{n-k}} - f^{\CC}(\xv) \right| \nonumber
\end{align}
where in Equation \eqref{eq:modC} we used that each term of the sum is constant on $\vec{x}+\CC$.
\end{proof}

As a  rewriting of Corollary \ref{coro:FB} we get the following proposition that upper-bounds $\Delta(u,f^{\CC})$, namely:

\begin{proposition}\label{propo:FBSDCod} Let $\CC$ be an $\lbrack n,k\rbrack$-code and $f$ be a radial distribution on $\F_2^n$.
We have
		$$
		\Delta\left(u, f^{\CC}\right) \leq  2^{n} \; \sqrt{\sum_{t = \dmin(\dual{\CC})}^{n} \Neq{t}{\dual{\CC}}|\widehat{f}(t)|^{2}}
		$$
where by abuse of notation we denote by $\widehat{f}(t)$ the common value of $\widehat{f}$ on vectors of weight $t$. 
\end{proposition}

\begin{proof}We have that $\CC$ is a closed subgroup of $\F_{2}^{n}$ with associated Haar measures:
	$$
	\mu_{\F_{2}^{n}} = \frac{1}{2^{n}} \; \sharp \quad \mbox{and} \quad \mu_{\F_{2}^{n}/\CC} = \frac{2^{k}}{2^{n}}\;\sharp
	$$
	for which we can apply Corollary \ref{coro:FB}. 
	Let $a \eqdef 2^{n-k} u$ and $b \eqdef 2^{n} f$. First, it is clear that $a : \vec{x}\in\F_{2}^{n}/\CC \mapsto 1$ and that 
	$$
	\int_{\F_{2}^{n}} b \;d\mu_{\F_{2}^{n}} = \frac{1}{2^{n}}\sum_{\vec{x}\in\F_{2}^{n}} 2^{n}f(\vec{x}) = 1 = \mu_{\F_{2}^{n}/\CC}(\F_{2}^{n}/\CC)
	$$
	where we used that $f$ is a distribution. Therefore we can apply Corollary \ref{coro:FB} with functions $a$ and $b$. Furthermore, $b^{|\CC} = 2^{n} f^{|\CC} = 2^{n-k}f^{\CC}$ by definition of $f^{\CC}$. We get the following computation:
	\begin{align}
		\| a - b^{|\CC} \|_{1} &= \| a - 2^{n-k} f^{\CC}\|_{1} \nonumber \\
		&= \sum_{\vec{x}\in\F_{2}^{n}/\CC} \left| 1 - 2^{n-k} f^{\CC}(\vec{x}) \right| \; \frac{1}{2^{n-k}} \nonumber\\
		&= \sum_{\vec{x}\in\F_{2}^{n}/\CC} \left| \frac{1}{2^{n-k}} - f^{\CC}(\vec{x}) \right| \nonumber\\
		&= 2\; \Delta(u,f^{\CC}) \ . \label{eq:toApplyCoro}
	\end{align}
	To conclude the proof it remains to apply Corollary \ref{coro:FB} with Equation \eqref{eq:toApplyCoro} and then to use that $f$ is radial and therefore also $\widehat{f}$. 
\end{proof}

Our upper-bound of Proposition \ref{propo:FBSDCod} involves the weight distribution of the code $\dual{\CC}$, namely $(\Neq{t}{\dual{\CC}})_{t \geq \dmin(\dual{\CC})}$. To understand how our bound behaves for a given distribution $f$, we will start (in the following subsection) with the case of random codes. The expected value for $\Neqs{t}$ is well known in this case. This will lead us to estimate our bound on almost all codes and  gives us some hints about the best distribution to choose for our smoothing bound in the worst case (which is the case that we treat in Subsection \ref{subsec:smoothFCode}).

\subsection{Smoothing Random Codes.} The probabilistic model $\crand{n}{k}$ that we use for our random code of length $n$ 
is defined by sampling uniformly at random a generator matrix $\vec{G} \in \F_{2}^{k\times n}$ for it, \textit{i.e.}
$$
\CC = \left\{ \vec{m}\vec{G} \mbox{ : } \vec{m} \in \F_{2}^{k} \right\}.
$$ 
It is straightforward to check that the expected number of codewords of weight $t$ in the dual $\dual{\CC}$ is given by:

\begin{fact}\label{fact:RCode}For $\CC$ chosen according to $\crand{n}{k}$
	\begin{equation*}
		\mathbb{E}_{\CC}(\Neq{t}{\dual{\CC}}) = \frac{\binom{n}{t}}{2^k}.
	\end{equation*}
\end{fact}

This estimation combined with Proposition \ref{propo:FBSDCod} enables us to upper-bound $\mathbb{E}_{\CC}\left(\Delta( u,f^{\CC})\right)$.

\begin{proposition}\label{propo:Rcode}
We have:
	\begin{equation}\label{bound:RandCodes}
	\mathbb{E}_{\CC}\left(\Delta(u,f^{\CC})\right) \leq 2^{n}\; \sqrt{ \sum_{t > 0} \frac{\binom{n}{t}}{2^{k}}\; |\widehat{f}(t)|^{2} } .
	\end{equation} 
\end{proposition}

\begin{proof} 
By using Proposition \ref{propo:FBSDCod}, we obtain:
\begin{align*}
	\mathbb{E}_{\CC}\left(\Delta(u,f^{\CC})\right) &\leq \mathbb{E}_{\CC}\left(2^{n} \; \sqrt{\sum_{t = \dmin(\dual{\CC})}^{n} \Neq{t}{\dual{\CC}}|\widehat{f}(t)|^{2}} \right) \\
	&\leq  2^{n}\; \sqrt{\mathbb{E}_{\CC}\left( \sum_{t = \dmin(\dual{\CC})}^{n} \Neq{t}{\dual{\CC}}|\widehat{f}(t)|^{2} \right)} \quad (\mbox{Jensen's inequality})\\
	&=  2^{n}\; \sqrt{ \sum_{t > 0} \frac{\binom{n}{t}}{2^{k}}\; |\widehat{f}(t)|^{2} }  
\end{align*}
where in the last line we used the linearity of the expectation and Fact \ref{fact:RCode}.  
\end{proof}

It remains now to choose the distribution $f$. A natural choice in code-based cryptography is the uniform distribution $\unifs{w}$ over the sphere $\Sc_w$ of radius $w$ centered around $\mathbf{0}$.

\noindent
{\bf \noindent Uniform Distribution over a Sphere.} 
The Fourier transform of 
$\unifs{w}$ 
is intimately connected to Krawtchouk polynomials. 
The Krawtchouk polynomial of order $n$ and degree $w\in \{ 0,\dots,n\}$ is defined as
$$
K_{w}(X;n) \eqdef \sum_{j=0}^{w} (-1)^j \binom{X}{j} \binom{n-X}{w-j}. 
$$ 
To simplify notation, since $n$ is clear here from context, we will drop the dependency on $n$ and simply write $K_w(X)$. The following fact allows to relate $K_{w}$ with $\widehat{\unifs{w}}$ (see for instance \cite[Lem. 3.5.1, \S 3.5]{L99})

\begin{fact}\label{fact:Kraw}
	For any $\vec{y} \in \mathcal{S}_{t}$,
	\begin{equation}
		\sum_{\vec{e}\in\mathcal{S}_{w}} (-1)^{ \vec{y}\cdot \vec{e}} = K_{w}(t).
	\end{equation}
\end{fact}

This leads us to 
$$
\widehat{\unifs{w}}(\vec{x}) = \frac{1}{2^{n}} \; K_{w}(|\vec{x}|)\bigg/\binom{n}{w}.$$
By plugging this in Equation \eqref{bound:RandCodes} of Proposition \ref{propo:Rcode} we obtain 
\begin{equation}
	\mathbb{E}_{\CC}\left( \Delta( u,	\unifs{w}^{\CC})\right) \leq \sqrt{\sum_{t > 0} \frac{\binom{n}{t}}{2^{k}} \left( \frac{K_{w}(t)}{\binom{n}{w}}\right)^{2}}.
\end{equation}

The above sum can be upper-bounded by observing that $\left( K_{w}/\sqrt{\binom{n}{w}}\right)_{0\leq w \leq n}$ is an orthonormal basis of functions $f:\{0,1,\cdots,n\}\rightarrow \C$ for the inner product  $\langle f,g \rangle_{\textup{rad}} \eqdef \sum_{t=0}^{n} f(t)\overline{g(t)} \binom{n}{t}/2^{n}$. It can be viewed as the standard inner product between radial functions over $\F_2^n$. In particular, $\sum_{t=0}^{n} \frac{K_{w}(t)^{2}}{\binom{n}{w}} \; \frac{\binom{n}{t}}{2^{n}} = 1$ \cite[Corollary 2.3]{L95}. 
Therefore, for random codes we obtain the following proposition
\begin{proposition}\label{propo:BUnifRandCase} We have for random $\CC$ chosen according to $\crand{n}{k}$
\begin{equation} \label{eq:BUnifRandCase}
	\mathbb{E}_{\CC}\left( \Delta(u,\unifs{w}^{\CC}) \right) \leq \sqrt{2^{n-k}\bigg/\binom{n}{w}}.
\end{equation}  
\end{proposition} 
In other words, if one wants to smooth a random code with target distance $2^{-\Omega(n)}$ via the uniform distribution over a sphere, one has to choose its radius $w\leq n/2$ 
such that $\binom{n}{w} = 2^{\Omega(n)} \; 2^{n-k}$. It is readily seen that for fixed code rate $R \eqdef \frac{k}{n}$, choosing any fixed ratio $\omega \eqdef \frac{w}{n}$ such that $\omega >\ogv(R)$ is enough, where $\ogv(R)$ corresponds to the asymptotic relative Gilbert-Varshamov (GV) bound
\[
	\omega_{\textup{GV}}(R) \eqdef h^{-1}(1-R) \ ,
\]
with $h^{-1}:[0,1] \to [0,1/2]$ being the inverse of the binary entropy function $h(p) = -p\log_2(p)-(1-p)\log_2(1-p)$. The GV bound $\ogv(R) $ appears ubiquitously in the coding-theoretic literature: amongst other contexts, it arises as the (expected) relative minimum distance of a random code of dimension $Rn$, or as the maximum relative minimum error weight for which decoding over the binary symmetric channel can be successful with non-vanishing error probability.

This value of radius $n \omega_{\textup{GV}}(R)$ is optimal: clearly, the support size of an error distribution smoothing a code $\CC$ must exceed $\sharp\F_2^n/\CC$. Thus, we cannot expect to smooth a code $\CC$ with errors in the sphere $\mathcal{S}_{w}$ if its volume is smaller than $2^{n-k} = \sharp\F_{2}^{n}/\CC$.

Therefore the uniform distribution over a sphere is optimal for random codes. By this, we mean that it leads to the smallest amount of possible noise (when it is concentrated on a ball) to smooth a random code. Notice that we obtained this result after applying the chain of arguments Cauchy-Schwarz, Parseval and Poisson to bound the statistical distance.

{\bf \noindent About the original chain of arguments of Micciancio and Regev.} It can be verified that by coming back to the original steps of \cite{MR07,ADRS15}, namely the Poisson summation formula and then the triangle inequality, we would obtain
\begin{equation}\label{eq:PFTI}
\Delta\left(u, f^{\CC} \right) \leq 2^{n} \sum_{t \geq \dmin(\dual{\CC})} \Neq{t}{\dual{\CC}} |\widehat{f}(t)| .
\end{equation} 
By using that $a^{2}+b^{2} \leq (a+b)^{2}$ (when $a,b \geq 0$) we see that our bound (Proposition \ref{propo:FBSDCod}) is sharper. It turns out that our bound is exponentially sharper for random codes (and even in the worst case) when choosing $f$ as the uniform distribution over a sphere of radius $w$, namely $f = \unifs{w}$. In this case the Micciancio-Regev argument yields the following computation
\begin{align}
	\mathbb{E}_{\CC}\left(\Delta\left( u, \unifs{w}^{\CC} \right)  \right) &\leq \mathbb{E}_{\CC}\left(  \sum_{t \geq \dmin(\dual{\CC})} \Neq{t}{\dual{\CC}} \; \frac{|K_{w}(t)|}{\binom{n}{w}} \right)  \nonumber\\ 
	&= \sum_{t > 0} \frac{\binom{n}{t}}{2^{k}} \; \frac{|K_{w}(t)|}{\binom{n}{w}} . \label{eq:BUnifMRRandCase}
\end{align}
To carefully estimate this upper-bound (and to compare with \eqref{eq:BUnifRandCase}) we are going to use the following proposition, which gives the asymptotic behaviour of $K_{w}$ (see for instance \cite{IS98,DT17}).

\begin{proposition}\label{prop:expansion}
	Let $n,t$ and $w$ be three positive integers.
	We set $\tau \eqdef \frac{t}{n}$, $\omega = \frac{w}{n}$ and $\omega^{\perp} \eqdef 1/2 - \sqrt{\omega(1-\omega)}$.
	We assume $w \leq n/2$. Let $z \eqdef \frac{1-2\tau  - \sqrt{D}}{2 (1-\omega)}$ where
	$D \eqdef \left(1- 2 \tau \right)^2-4 \omega(1-\omega)$.
	In the case $\tau \in ( 0,\omega^\perp)$, 
	\begin{equation*}
		K_{w}(t) = O\left( 2^{n(a(\tau,\omega)+o(1))}\right) \quad \mbox{where} \quad a(\tau,\omega) \eqdef \tau \log_{2}(1-z) + (1-\tau)\log_{2}(1+z) - \omega\log_{2}z .
\end{equation*}
	\item In the case $\tau \in (\omega^\perp,1/2)$, $D$ is negative, 
	and 
	\begin{equation*}
		K_{w}(t) = O\left( 2^{n(a(\tau,\omega)+o(1))} \right) \quad \mbox{where} \quad a(\tau,\omega) \eqdef \frac{1}{2}(1+h(\omega)-h(\tau)). 
\end{equation*}
\end{proposition}
We let,
\begin{align*}
\omega_0 &\eqdef \varlimsup_{n \rightarrow \infty}\left\{\frac{w}{n}: \;\sqrt{2^{n(1-R)}\bigg/\binom{n}{w}} \geq 1\right\},\\
\omega_1 &\eqdef \varlimsup_{n \rightarrow \infty}\left\{\frac{w}{n}: \;\sum_{t > 0} \frac{\binom{n}{t}}{2^{Rn}} \;\frac{|K_{w}(t)|}{\binom{n}{w}} \geq 1 \right\}.
\end{align*}
In Figure \ref{figure:compareBothApproach} we compare the asymptotic values of $\omega_0$ and $\omega_1$ as functions of $R$.
 Notice that $\omega_0 = \ogv(R)$. We see that $\omega_1$ is undefined for a rate $R < 1/2$. In other words, it is impossible to show that $\mathbb{E}_{\CC}\left( \Delta(u, \unifs{w}^{\CC}) \right) \leq 2^{-\Omega(n)}$ with the standard approach of \cite{MR07,ADRS15} when $R<1/2$. Furthermore, for larger rates (and sufficiently large $n$),  $\omega_0$ is much smaller than $\omega_1$. 
\begin{center}
	\begin{figure}[htb]
		\includegraphics[height=6cm]{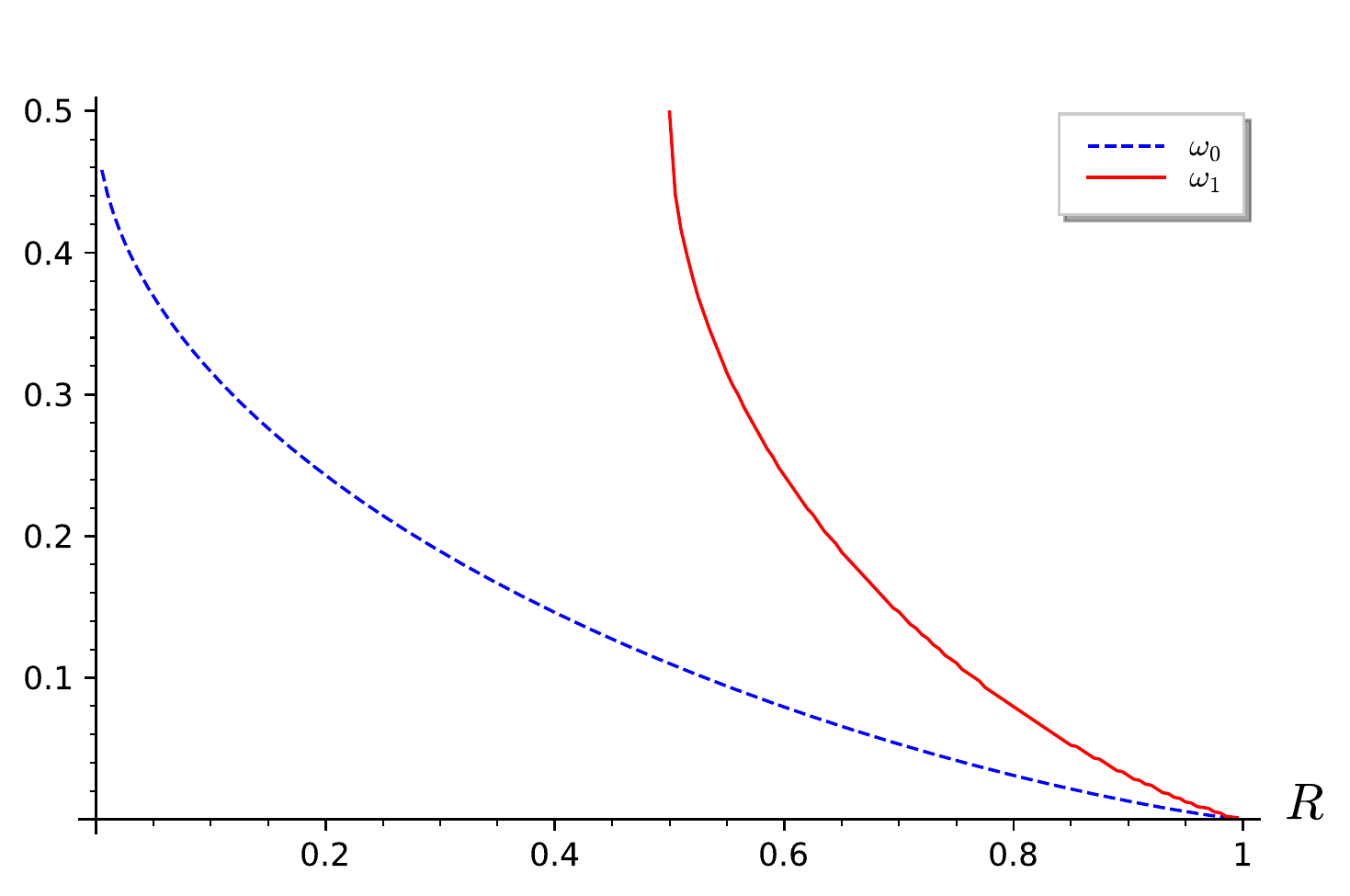}
		\caption{$\omega_{0}$  and  $\omega_{1}$ as functions of $R \eqdef \frac{k}{n}$.
\label{figure:compareBothApproach}}
	\end{figure}
\end{center}

\medskip
{\noindent \bf Bernoulli Distribution.} Another natural distribution to consider when dealing with codes is the so-called ``Bernoulli'' distribution $\fber{p}$, which is defined for $p\in[0,1/2]$ as
$$
\forall \vec{x}\in\F_{2}^{n}, \quad \fber{p}(\vec{x}) \eqdef p^{|\vec{x}|} (1-p)^{n-|\vec{x}|} .
$$
This choice leads to simpler computations compared to the uniform distribution over a sphere. For instance we have $\widehat{\fber{p}}(\vec{x}) = \frac{1}{2^{n}}(1-2p)^{|\vec{x}|}$. By plugging this in Equation \eqref{bound:RandCodes} of Proposition \ref{propo:Rcode} we obtain 
\begin{align} 
	\mathbb{E}_{\CC}\left( \Delta(u,\fber{p}^{\CC}) \right) &\leq \sqrt{\sum_{t > 0}\frac{\binom{n}{t}}{2^{k}}(1-2p)^{2t}} \nonumber \\ 
	&\leq \sqrt{\frac{1}{2^{k}} \; (1+(1-2p)^{2})^{n} } \label{eq:BUnifBerCase}
\end{align}

Thus, if one wants to smooth a random code at target distance $2^{-\Omega(n)}$ with the Bernoulli distribution, the above argument says that one has to choose $p > p_{0} \eqdef \frac{1}{2}\left(  1-\sqrt{2^{R}-1} \right)$ where $R = k/n$. 
As $\mathbb E_{\fber{p}}(|\vec x|) = pn$, it is meaningful to compare $p_{0}$ and $\omega_0$.
It is readily seen that $\omega_0=\ogv(R)=h^{-1}(1-R) < \frac{1}{2}\left(  1-\sqrt{2^{R}-1}\right) =p_0$.
In other words, this time the upper-bound given by Proposition  \ref{propo:Rcode} does not give what would be optimal, namely the Gilbert-Varshamov relative distance $\ogv(R)$, but a quantity which is bigger. However, it is expected that the average amount of noise to smooth a random code is the same in both cases, since a Bernoulli distribution of parameter $p$ is extremely concentrated over words of Hamming weight $pn$ and that therefore  $\Delta(u, \fber{p}^{\CC}) \approx \Delta(u, \unifs{pn}^{\CC})$. This suggests that Proposition 
 \ref{propo:Rcode} is not tight in this case. This is indeed the case, we can prove that we can smooth a random code with the Bernoulli noise as soon as $p > \ogv(R)$. This follows from the following proposition.

 \begin{restatable}{proposition}{BerVSUnif}\label{propo:BerVSUnif}  Let $\varepsilon > 0$ and $p\in [0,1/2]$. Then,
 	$$
 	\Delta(u,f^{\CC}_{\textup{ber},p}) \leq \sum_{r=(1-\varepsilon)np}^{(1+\varepsilon)np} \Delta(u,\unifs{r}^{\CC})  + 2^{-\Omega(n)} .
 	$$

 \end{restatable}

\begin{proof} 
	See Appendix \ref{app:BerVSUnif}.
\end{proof}

This proposition shows that if one wants $\Delta(u,f^{\CC}_{\textup{ber},p}) \leq 2^{-\Omega(n)}$ it is enough to have $\Delta(u,f^{\CC}_{\textup{unif},r}) \leq 2^{-\Omega(n)}$ for any $r \in \left[(1-\varepsilon)np,(1+\varepsilon)np\right]$. This can be achieved by choosing 
$\varepsilon$ and $p$ such that $(1-\varepsilon)p > \ogv(R)$.

To summarize this subsection we have the following theorem

\begin{theorem}Let $\CC$ be a random code chosen according to $\crand{n}{k}$, $R \eqdef \frac{k}{n}$.
Let $u$ (resp. $\unifs{\lceil pn \rceil}$) be the uniform distribution over $\F_{2}^{n}/\CC$ (resp. $\mathcal{S}_{w}$) and $\fber{p}$ be the Bernoulli distribution over $\F_{2}^{n}$ of parameter $p$. We have,
	$$
	\mathbb{E}_{\CC}\left( \Delta( u,	\unifs{\lceil pn \rceil}^{\CC})\right) \leq \; 2^{\frac{n}{2}\left( 1-R - h(p) + o(1)  \right)} \quad \mbox{and} \quad \mathbb{E}_{\CC}\left( \Delta( u,	\fber{p}^{\CC})\right) \leq 2^{\frac{n}{2}\left( 1-R - h(p) + o(1)  \right)}.
	$$ 
	In particular, $\mathbb{E}_{\CC}\left( \Delta( u,	\unifs{\lceil pn \rceil}^{\CC})\right) \leq 2^{-\Omega(n)}$ and $\mathbb{E}_{\CC}\left( \Delta( u,	\fber{p}^{\CC})\right) \leq 2^{-\Omega(n)}$ for any fixed $p > \ogv(R)$. 
\end{theorem}

\subsection{Smoothing a Fixed Code\label{subsec:smoothFCode}.} Our upper-bound on $\Delta(u,f^{\CC})$ given in Proposition \ref{propo:FBSDCod} involves the weight distribution of the dual of $\CC$, namely the $\Neq{t}{\dual{\CC}}$'s. To derive smoothing bounds on a fixed code our strategy will simply consist in using the best known upper bounds on the $\Neq{t}{\dual{\CC}}$'s. Roughly speaking, these bounds show that $\Neq{t}{\dual{\CC}} \leq \binom{n}{t}2^{-Kn}$ for some constant $K$ which is function of $\dmin(\dual{\CC})$.  
\newline

{\bf \noindent Notation.}
	Let $\delta \in (0,1/2)$ and $\delta \leq \tau \leq 1$,
	\begin{equation}\label{eq:bDeltaTau}
		b(\delta,\tau) \eqdef \mathop{\overline{\lim}}\limits_{n \to \infty} \mathop{\max}\limits_{\CC} \left\{ \frac{1}{n}\log_{2} \Neq{\lfloor\tau n\rfloor}{\CC} \right\} 
	\end{equation} 
	where the maximum is taken over all codes $\CC$ of length $n$ and minimum distance $\geq \delta n$.

We recall (or slightly extend) results taken from \cite{ABL01}:
\begin{restatable}{proposition}{PropoABL}\label{propo:ABL}
		Let $\delta \in (0,1/2)$ and $\delta^{\perp} \eqdef 1/2 - \sqrt{\delta(1-\delta)}$. For any $\delta \leq \tau \leq 1$
	\begin{equation}
		b(\delta,\tau) \leq c(\delta,\tau) \eqdef   \left\{
		\begin{array}{ll}
				h(\tau) + h\left( \delta^{\perp} \right) - 1 & \mbox{if } \tau \in [\delta,1-\delta], \\
			 	2\left( h(\delta^{\perp}) - a(\tau,\delta^{\perp}) \right) & \mbox{otherwise,}
		\end{array}
		\right.
	\end{equation}
	where $a(\cdot,\cdot)$ is defined in Proposition \ref{prop:expansion}.

\end{restatable}

\begin{proof}See Appendix \ref{app:proofPropoABL}.
\end{proof}

\begin{proposition}[{\cite[Proposition 4]{ABL01}}]\label{propo:2LPB} Let $\delta_{\textup{JSB}} \eqdef \left( 1 - \sqrt{1-2\delta}\right)/2$ and 
	$$
	\tau_{0} \eqdef \mathop{\argmin}\limits_{\delta_{\textup{JSB}} \leq \alpha \leq 1/2} 1- h(\alpha) + R_{1}(\alpha,\delta)
	$$
	where 
	$$
	R_{1}(\tau,\delta) \eqdef h\left(\frac{1}{2}\left( 1 - \sqrt{1- \left(\sqrt{4\tau(1-\tau)-\delta(2-\delta)} - \delta\right)^{2}}  \right)   \right).
	$$  
	 For any $\delta \leq \tau \leq 1$ 
	\begin{equation} 
		b(\delta,\tau) \leq  d(\delta,\tau) \eqdef \left\{
	\begin{array}{ll}
h(\tau) - h(\tau_{0}) + R_{1}(\tau_{0},\delta)  & \mbox{if } \tau \in (\delta_{\textup{JSB}},1-\delta_{\textup{JSB}}) \mbox{ and } \tau_{0} \leq \tau,  \\
			R_{1}(\tau,\delta)  & \mbox{if } \tau \in (\delta_{\textup{JSB}},1-\delta_{\textup{JSB}}) \mbox{ and } \tau_{0} > \tau,  \\
		 0 & \mbox{otherwise.}
	\end{array}
	\right.
	\end{equation}  

\end{proposition}

Both of these bounds are derived from ``linear programming arguments'' which were initially used to upper-bound the size of a code given its minimum distance. Proposition \ref{propo:ABL} is an extension of \cite[Theorem 3]{ABL01} in the case of linear codes, in particular we give an upper-bound for any $\tau \in [\delta,1]$ (and not for only $\tau \in [\delta,1/2]$). The proof is in the appendix. The second bound is usually called the {\em the second linear programming bound}. 
In terms of $\delta$ and $\tau$, Proposition \ref{propo:ABL} and \ref{propo:2LPB} are among the best (known) upper-bounds on $b(\delta,\tau)$. In the case where $0 \leq \delta \leq 0.273$, Proposition \ref{propo:2LPB} leads to better smoothing bounds compared to Proposition \ref{propo:ABL}.

\begin{remark}
	There exist many other bounds on $b(\delta,\tau)$, like \cite[Theorem 8]{ACKL05} which holds only for linear codes or \cite[Theorem 7]{ACKL05}. However for our smoothing bounds, Propositions \ref{propo:ABL} and \ref{propo:2LPB} lead to the best results, partly because these are the best bounds on the number of codewords of Hamming weight close to the minimum distance of the code.  
\end{remark}

We draw in Figures \ref{figure:Bounds01} and \ref{figure:Bounds035} the bounds of Propositions \ref{propo:ABL} and \ref{propo:2LPB} as function of $\tau \in [\delta,1]$ for a couple values of $\delta$.

\begin{center}
	\begin{figure}
		\includegraphics[height=6cm]{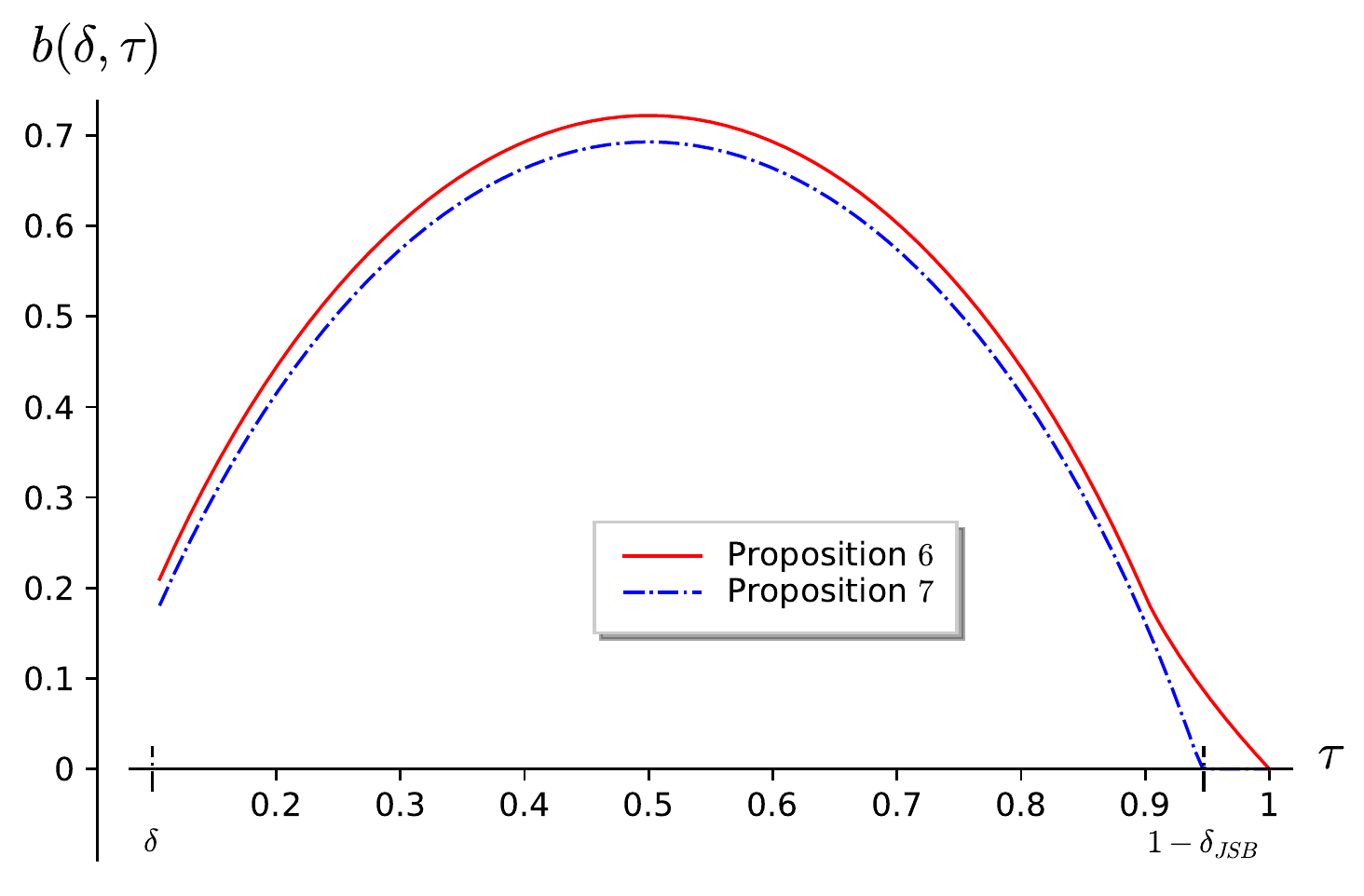}
		\caption{Bounds of Propositions \ref{propo:ABL} and \ref{propo:2LPB} on $b(\delta,\tau)$ as function of $\tau \in [\delta,1]$ for $\delta = 0.1$.
			\label{figure:Bounds01}}
	\end{figure}
\end{center}

\begin{center}
	\begin{figure}
		\includegraphics[height=6cm]{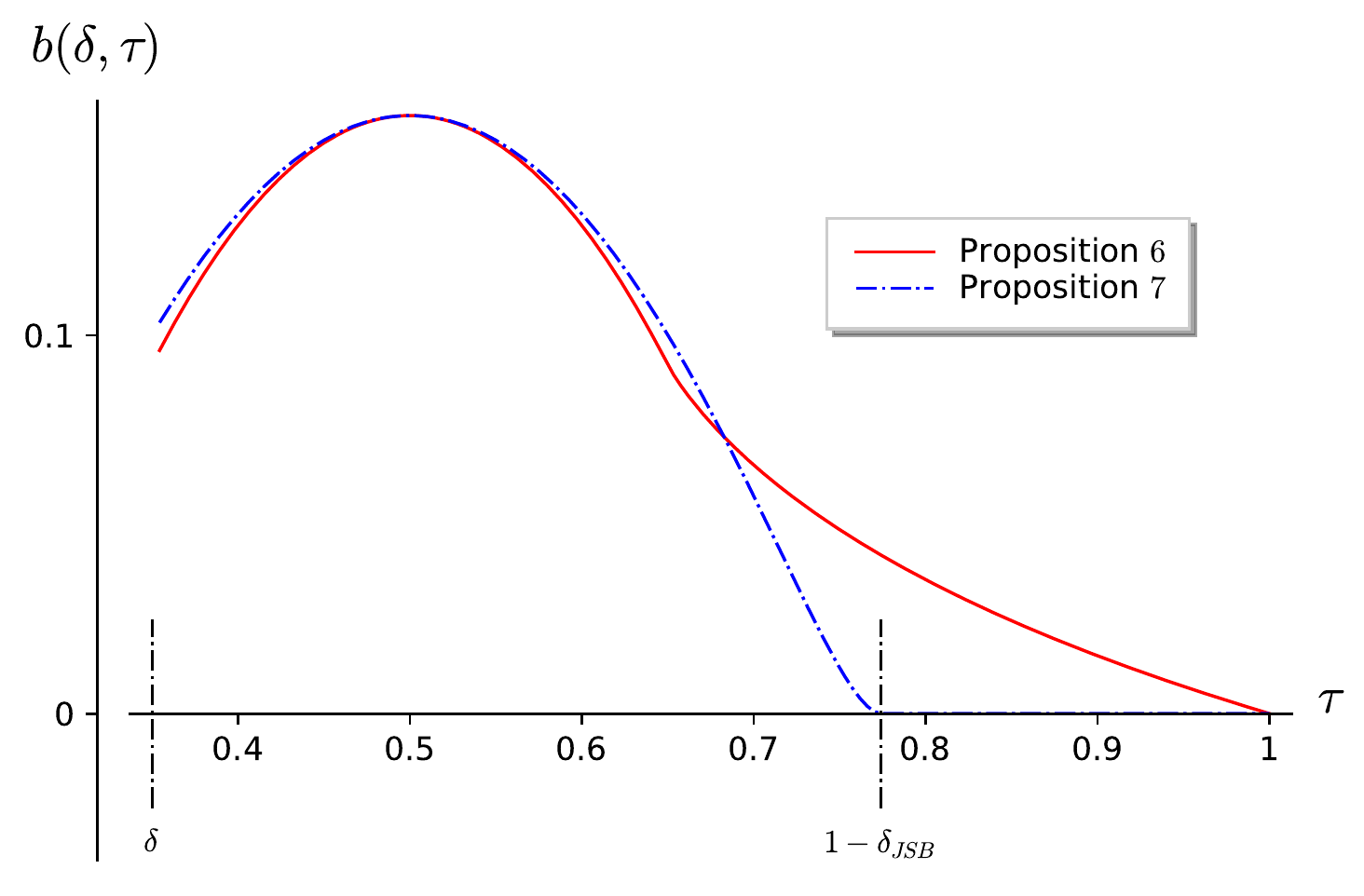}
		\caption{Bounds of Propositions \ref{propo:ABL} and \ref{propo:2LPB} on $b(\delta,\tau)$ as function of $\tau \in [\delta,1]$ for $\delta = 0.35$. 
			\label{figure:Bounds035}}
	\end{figure}
\end{center}

Equipped with these bounds we are ready to give our smoothing bounds for codes in the worst case, namely for a fixed code. Our study with random codes gave a hint that the choice of the uniform distribution over a sphere could give better results than the Bernoulli distribution. However, as we will show now, the distribution on a sphere forces us to assume that no codewords of large weight belong to the dual $\dual{\CC}$ when we want to smooth $\CC$. It corresponds to the hypothesis of balanced-codes made in \cite{BLVW19} to obtain a worst-to-average case reduction. We would like to avoid making this assumption as nothing forbids large weight vectors from belonging to a fixed code. Fortunately, as we will later show, we can avoid making this hypothesis while still keeping the advantages of the uniform distribution over a sphere. 
\newline

{\bf \noindent Impossibility to smooth a code whose dual is not balanced with the uniform distribution over a sphere.} 
It is readily seen that in the case where the dual code $\dual{\CC}$ is not balanced, meaning that it contains the all-one vector (and therefore that the dual weight distribution is symmetric: $\Neq{w}{\dual{\CC}}=\Neq{n-w}{\dual{\CC}}$ for any $w \in \{0,\cdots,n\}$ when the codelength is $n$), then it is impossible to smooth it with the uniform distribution $\unifs{w}$ over a sphere. Indeed, this implies that all codewords of $\CC$ have an even Hamming weight (they have to be orthogonal to the all-one vector). The parity of the Hamming weights of vectors in a coset ({\em i.e.} in the class of representatives of some element in $\F_{2}^{n}/\CC$) will be the same. Therefore, half of the cosets cannot be reached when periodizing $\unifs{w}$ over $\CC$.
\newline

{\bf Difficulty of using  Proposition \ref{propo:FBSDCod} for proving smoothness of the uniform distribution if the dual has large weight codewords.} Even in the case where the dual is balanced, difficulties can arise if we want to use Proposition \ref{propo:FBSDCod} for proving smoothness of the uniform distribution over a sphere when the dual has large weight codewords. First of all, the fact that it contains 
the all-one codeword also reflects in the upper-bound of Proposition \ref{propo:FBSDCod}. Recall that $\widehat{\unifs{w}}(\vec{x}) = \frac{1}{2^{n}}K_{w}(|\vec{x}|)/\binom{n}{w}$ and that we have $K_{w}(n) = (-1)^{w}\binom{n}{w}$ (see Fact \ref{fact:Kraw}). Therefore, when the full weight vector belongs to $\dual{\CC}$, our upper-bound on $\Delta(u,\unifs{w}^{\CC})$ of Proposition \ref{propo:FBSDCod} cannot be smaller than $1$. Furthermore, even if the dual does not contain the all-one codeword, codewords of weight say $t=n - O(\log n)$ also give a non-negligible contribution to the upper-bound of Proposition \ref{propo:FBSDCod}: the contribution is a polynomial $n^{-O(1)}$.
\newline

{\bf Difficulty of using  Proposition \ref{propo:FBSDCod} for proving smoothness of  the ``discrete walk distribution'' if the dual has large weight codewords.} 
Other  meaningful distributions in the cryptographic context display the same problem as the uniform distribution concerning the difficulty of applying Proposition \ref{propo:FBSDCod} to them if the dual contains large weight codewords. This applies to the 
discrete time random walk distribution $f_{\textup{RW},t}$ introduced in \cite{BLVW19} for worst-to-average case reductions. 
The authors were only able to prove smoothness of this distribution if the dual code has no small {\em and no large} weight codewords.
This distribution is given by
$$
\forall \vec{x}\in\F_{2}^{n},\quad f_{\textup{RW},w}(\vec{x}) \eqdef \mathbb{P}\left( \sum_{i=1}^{w} \vec{e}_{u_{i}} = \vec{x} \right) 
$$
where the $u_{i}$'s are independently and uniformly drawn at random in $\{1,\dots,n\}$ and $\vec{e}_{j}$ denotes the $j$-th canonical basis vector. Recall that \cite{BLVW19}
$$
\widehat{f_{\textup{RW},w}(\vec{y})} = \frac{1}{2^{n}} \; \left( 1-2\frac{|\vec{y}|}{n}\right)^{w}.
$$
Therefore, $\widehat{f_{\textup{RW},w}}(\vec{y}) = \frac{1}{2^{n}} (-1)^{w}$ when $|\vec{y}| = n$, as for the Fourier transform of the uniform distribution over a sphere, showing that $f_{\textup{RW},w}$ cannot smooth a code when the full weight vector belongs to its dual. 
In summary, {\em a direct application} of Proposition \ref{propo:FBSDCod} is quite  unsatisfactory for these distributions $\unifs{w}$ and $f_{\textup{RW},w}$. If we are willing to also make an assumption on the largest weight of a codeword, then certainly a direct application of Proposition \ref{propo:FBSDCod} is able to provide meaningful smoothing bounds for them. Indeed, the following theorem is obtained by just combining Propositions \ref{propo:FBSDCod},  \ref{propo:ABL} and \ref{propo:2LPB}.

\begin{theorem}\label{theo:smoothingBoundsUnifRW}
	Let $\CC$ be a binary linear code of length $n$ and $\omega\in(0,1)$. Suppose that $\dmin(\dual{\CC}) = \dual{\delta} n$ and that $\dual{\CC}$ has no element of Hamming weight $\geq \beta n$ for some $\beta \in (\dual{\delta},1)$. We have
		$$
		\frac{1}{n} \log_{2} \Delta\left(u,\unifs{\omega n}^{\CC}\right)  \leq  \mathop{\max}\limits_{\dual{\delta} \leq \tau \leq \beta} \left\{ \frac{1}{2} \min\left\{ c(\dual{\delta},\tau),d(\dual{\delta},\tau) \right\} + 
		 a(\omega,\tau)   \right\} - h(\omega)
	$$
	$$
	\frac{1}{n} \log_{2} \Delta\left(u,f_{\textup{RW},\omega n}^{\CC}\right)  \leq  \mathop{\max}\limits_{\dual{\delta} \leq \tau \leq \beta} \left\{ \frac{1}{2} \min\left\{ c(\dual{\delta},\tau),d(\dual{\delta},\tau) \right\}  + \omega\log_{2}\left(1-2\tau\right) \right\}
	$$
	
		where $a(\cdot,\cdot)$, $c(\cdot,\cdot)$ and $d(\cdot,\cdot)$ are defined respectively in Propositions \ref{prop:expansion}, \ref{propo:ABL} and \ref{propo:2LPB}. 
\end{theorem}

{\bf \noindent Avoiding making an assumption on the largest dual codeword: the case of the Bernoulli distribution.}
Even if the Bernoulli distribution has some drawbacks compared to the uniform distribution over a sphere, when applying Proposition \ref{propo:FBSDCod} with random codes, it has however a nice property concerning the large weight codewords: the large weight dual codewords have a negligible contribution  in the upper-bound of Proposition \ref{propo:FBSDCod}.
To see this let us first recall that
\begin{equation} 
\widehat{\fber{p}}(\vec{x}) = \frac{1}{2^{n}} \; (1-2p)^{|\vec{x}|}.
\end{equation}
Therefore, by Proposition \ref{propo:FBSDCod} we have
\begin{equation}\label{eq:BerToBound}
	\Delta(u,\fber{p}^{\CC}) \leq \sqrt{\sum_{t = \dmin(\dual{\CC})}^{n} \Neq{t}{\dual{\CC}} (1-2p)^{2t}}.
\end{equation}
On the other hand, we have the following lemma which shows that large weight codewords can only have an exponentially small contribution to the above upper-bound.
\begin{lemma}\label{lemma:LCodewords}
	Let $\CC$ be a linear code of length $n$ and let $t > n-\dmin(\CC)/2$. There is at most one codeword $\vec{c}$ of weight $t$. 
\end{lemma}

\begin{proof}
	Suppose by contradiction that there exists two distinct codewords $\vec{c},\vec{c}'\in\CC$ of Hamming weight $t$. By using the triangle inequality we obtain (where $\vec{1}$ denotes the all-one vector)
	\begin{align*}
		|\vec{c} - \vec{c}'| & \leq |\vec{c} - \vec{1}| + |\vec{1} - \vec{c}'| \\
		&= 2\left( n-t \right) \\
		&< \dmin(\CC)
	\end{align*}
	which contradicts the fact that $\CC$ has minimum distance $\dmin(\CC)$. 
\end{proof}

Therefore, using Lemma \ref{lemma:LCodewords} in Equation \eqref{eq:BerToBound} gives for $p\in(0,1/2]$,
\begin{equation}\label{eq:BwithBer}
	\Delta(u,\fber{p}^{\CC}) \leq \sqrt{\sum_{t = \dmin(\dual{\CC})}^{n-\dmin(\dual{\CC})/2} \Neq{t}{\dual{\CC}} (1-2p)^{2t}} + 2^{-\Omega(n)}.
\end{equation}

In other words, large weight dual codewords (if they exist) have only an exponentially small contribution to our smoothing bound with the Bernoulli distribution. In principle, we could plug in Equation \eqref{eq:BwithBer} bounds on the $\Neq{t}{\dual{\CC}}$'s given in Propositions \ref{propo:ABL} and \ref{propo:2LPB}. 
We will improve on the bounds obtained in this way by truncating the Bernoulli distribution, then \\
 \begin{itemize}
 	\item[$(i)$] prove that by appropriately truncating both distributions have the same smoothness property,
 	\item[$(ii)$] show that the truncated distribution has the same nice properties with respect to large weights,
 	\item[$(iii)$] show that we can apply Proposition \ref{propo:FBSDCod} to the truncated distribution and get appropriate smoothness properties.
 \end{itemize}
We obtain in this way:

	\begin{restatable}{theorem}{thFinalCode}\label{theo:finalUBSD} 
	Let $\CC$ be a binary linear code of length $n$ and $p \in (0,1/2]$ such that  $\dmin(\dual{\CC}) \geq \dual{\delta} n$ for some $\dual{\delta}\in[0,1]$. We have asymptotically,
\begin{multline*}
			\frac{1}{n} \log_{2} \Delta\left(u,\fber{p}^{\CC}\right)  \leq  \mathop{\max}\limits_{\dual{\delta} \leq \tau \leq 1 - \dual{\delta}/2} \{ \frac{1}{2} \min \left\{c(\dual{\delta},\tau),d(\dual{\delta},\tau) \right\} + \\
			 \mathop{\max}\limits_{(1-\varepsilon)p \leq \lambda \leq (1+\varepsilon)p} \left\{ \lambda\log_{2}p + (1-\lambda)\log_{2}(1-p) + a(\lambda,\tau) \right\}  \}  + O\left(\frac{1}{n}\right)
	\end{multline*}
	where $a(\cdot,\cdot)$, $c(\cdot,\cdot)$ and $d(\cdot,\cdot)$ are defined respectively in Propositions \ref{prop:expansion}, \ref{propo:ABL} and \ref{propo:2LPB}. 
	\end{restatable}

	\begin{proof} 
		See Appendix \ref{app:proofThFinalCode}.
	\end{proof}

	Let $i\in\{0,1\}$ and $p_{i}$ be the smallest $p\in(0,1/2]$ that enables to reach $\Delta\left(u,\fber{p}^{\CC}\right) \leq 2^{-\Omega(n)}$ with
	\begin{itemize}
		\item Theorem \ref{theo:finalUBSD} when $i=0$,

		\item  Equation \eqref{eq:BwithBer} and Propositions \ref{propo:ABL}, \ref{propo:2LPB} when $i =1$. 
	\end{itemize}
	
	In Figure \ref{figure:compBerTrunc} we compare the smallest $p$ that enables one to reach $\Delta\left(u,\fber{p}^{\CC}\right) \leq 2^{-\Omega(n)}$ with Equation \eqref{eq:BwithBer} and with Theorem \ref{theo:finalUBSD}. 
	As we can see  Theorem \ref{theo:finalUBSD} leads to significantly better bounds. Furthermore,  it turns out that $p_{0}n$ is roughly equal to the smallest radius $w$ such that $\Delta(u,\unifs{w}^{\CC}) \leq 2^{-\Omega(n)}$ if we had supposed that no codewords of weight 
$> n - \dmin(\dual{\CC})$ belong to $\dual{\CC}$. 
In other words, our proof using the tweak of truncating the Bernoulli enables us to obtain a smoothing bound without the hypothesis of no dual codewords of large Hamming weight which is as good as with the uniform distribution over a sphere if we had made this assumption.

	\begin{center}
		\begin{figure}
			\includegraphics[height=6cm]{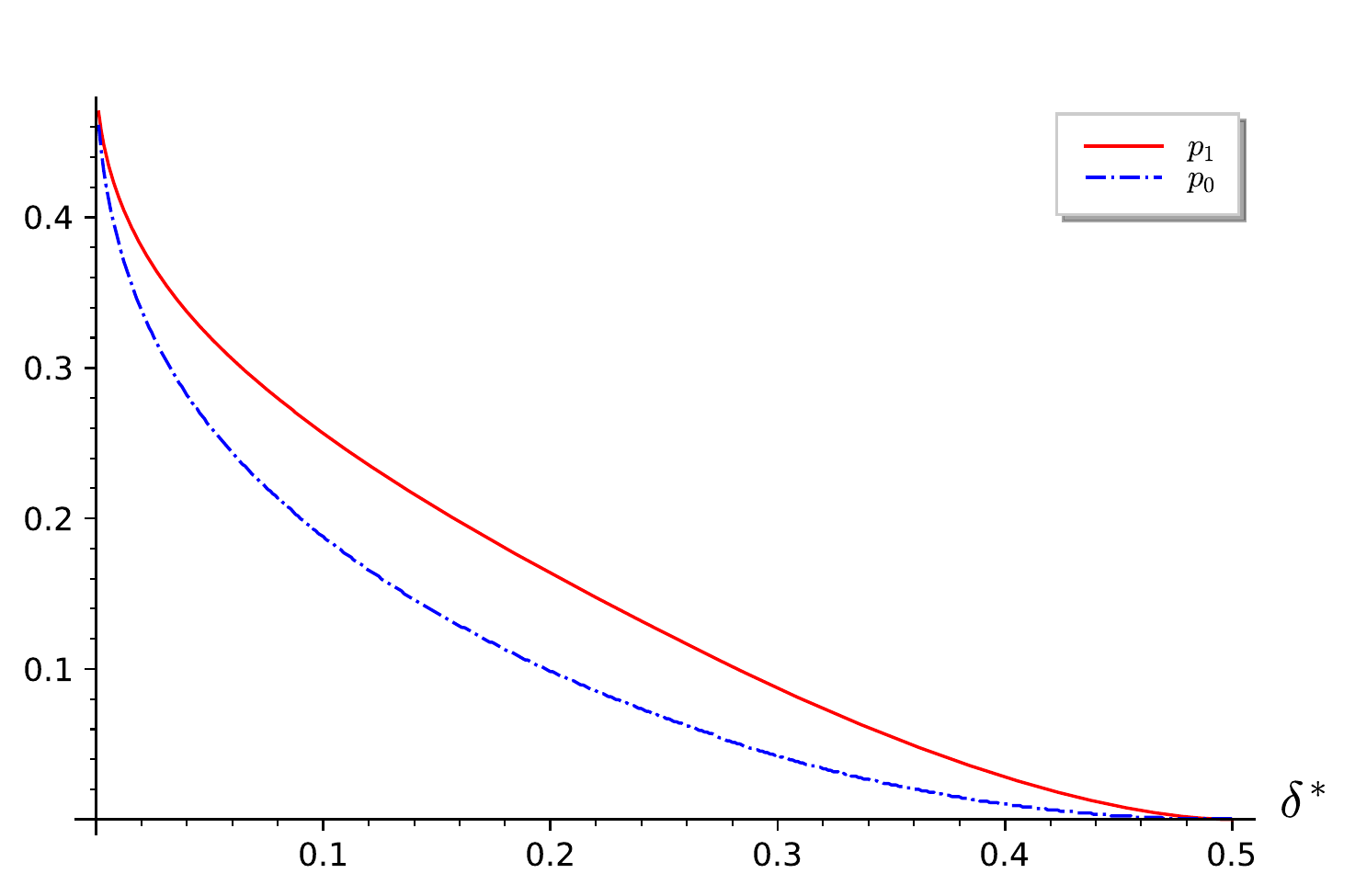}
			\caption{Smoothing bounds for a code $\CC$ as function of $\dual{\delta} \eqdef \dmin(\dual{\CC})/n$ via Theorem \ref{theo:finalUBSD} (for $\varepsilon = 10^{-2}$) and Equation \eqref{eq:BwithBer}  \label{figure:compBerTrunc}.}
		\end{figure}
	\end{center}

\section{Smoothing Bounds: Lattice Case}
\label{sec:SBLat} 

Given an $n$-dimensional lattice $\Lambda$ the aim of smoothing bounds is to give a non-trivial model of noise $\vec{e} \in \mathbb{R}^{n}$ for $(\vec{e}\mod \Lambda)\in\mathbb{R}^{n}/\Lambda$ (namely the reduction of $\vec{e}$ modulo $\Lambda$) to be uniformly distributed. Following Micciancio and Regev \cite{MR07}, the standard choice of noise is given by the Gaussian distribution, defined via
$$
\forall \vec{x}\in \mathbb{R}^{n},\quad D_{s}(\vec{x}) \eqdef \frac{1}{s^{n}}\;\rho_{s}(\vec{x}) \quad \mbox{where} \quad \rho_{s}(\vec{x}) \eqdef e^{-\pi(|\vec{x}|_{2}/s)^{2}} \ .
$$
The parametrization is chosen such that $s\sqrt{n/2\pi}$ is the standard deviation of $D_s$. Micciancio and Regev showed that when $\vec{e}$ is distributed according to $D_{s}$, choosing $s$ large enough enables $\vec{e} \mod \Lambda$ to be statistically close to the uniform distribution. 

However, following the intuition from the case of codes we will first analyze the case where $\vec e$ is sampled uniformly from a  Euclidean ball. Interestingly, just as with codes where our methodology led to stronger bounds when the uniform distribution over a sphere was used to smooth rather than the Bernoulli distribution, we will obtain better results when we work with the uniform distribution over a ball. Fortunately, using concentration of the Gaussian measure one can translate results from the case where $\vec e$ is uniformly distributed over a ball to the case that it is sampled according to $D_s$; see Proposition~\ref{propo:gauss-to-unif}. This is analogous to the translation from results for the uniform distribution over a sphere to the Bernoulli distribution for codes elucidated in Proposition~\ref{propo:BerVSUnif}.

For either choice of noise, to obtain a smoothing bound we are required to bound the statistical distance between the distribution of $\vec{e} \mod \Lambda$ if $\vec{e}$ has density $g$, and the uniform distribution over $\mathbb{R}^{n}/\Lambda$. It is readily seen that $\vec{e} \mod \Lambda$ has density $|\Lambda| g^{|\Lambda}$ which is defined as (see Definition \ref{def:perio} with the choice of Haar measures given in Table \ref{fig:groups})
$$
g^{|\Lambda}(\vec{x}) = \frac{1}{|\Lambda|}\; \sum_{\vec{y}\in\Lambda} g(\vec{x} + \vec{y}).
$$
{\noindent \bf Notation.} For any $g : \mathbb{R}^{n} \rightarrow \mathbb{C}$,
$$
g^{\Lambda} \eqdef |\Lambda| \; g^{|\Lambda}.
$$
In the following proposition we specialize Corollary~\ref{coro:FB} to the case of lattices.

	\begin{proposition}\label{propo:FBSDLat} Let $\Lambda$ be an $n$-dimensional lattice. Let $g$ be some density function on $\mathbb{R}^{n}$ and $v$ be the density of the uniform distribution over $\mathbb{R}^{n}/\Lambda$. We have
		$$
			\Delta\left(v,g^{\Lambda}\right) \leq  \frac{1}{2}\;\sqrt{\sum_{\vec x \in \dual{\Lambda}\setminus \{\vec 0\}} |\widehat{g}(\vec x)|^{2}} \ .
		$$
	\end{proposition} 

	We will restrict our instantiations to functions $g$ whose Fourier transforms are radial, that is, $\widehat{g}(\vec{x})$ depends only on the Euclidean norm of $\vec{x}$, namely $|\vec{x}|_{2}$.

	\subsection{Smoothing Random Lattices}
	As with codes, we begin our investigation of smoothing lattices by considering the random case. However, defining a ``random lattice'' is much more involved than the analogous notion of random codes. Fortunately for us, we can apply the Siegel version of the Minkowski-Hlawka theorem to conclude that there exists a random lattice model which behaves very nicely from the perspective of ``test functions''. We first state the technical theorem that we require.

	\begin{theorem} [Minkowski-Hlawka-Siegel] \label{theo:Mink-Hl-Si}
		On the set of all the lattices of covolume $M$ in $\R^n$ there exists a probability measure $\mu$ such that, for any Riemann integrable function $g(\vec{x})$ which vanishes outside some bounded region,\footnote{This statement holds for a larger class of functions. In particular it holds for our instantiation with the Gaussian distribution.}
		\[
			\mathop{\mathbb E}_{\Lambda \sim \mu}\left(\sum_{\vec x \in \Lambda\setminus\{\vec 0\}}g(\vec x)\right) = \frac{1}{M}\int_{\R^n}g(\vec x)d\vec x \ .
		\]
	\end{theorem}

	As intuition for the above theorem, consider the case that $g$ is the indicator function for a bounded, measurable subset $S \subseteq \R^n$. Then, Theorem~\ref{theo:Mink-Hl-Si} promises that the expected number of lattice points (other than the origin\footnote{Note that as $\vec 0 \in \Lambda$ with certainty, there is really no ``randomness'' for this event.}) in $S$ is equal to the volume of $S$ over the covolume of the lattice.
	
	{\bf \noindent Uniform Distribution over a Ball.} Let
	$$
	\gunif{w} \eqdef \frac{1_{\Bc_{w}}}{\vol w}
	$$ be the density of the uniform distribution over the Euclidean ball of radius $w$. Let us recall that $\vol w$ denotes the volume of any ball of radius $w$. From Theorem~\ref{theo:Mink-Hl-Si}, we may obtain the following proposition. This should be compared with Proposition~\ref{propo:BUnifRandCase}.
	
	\begin{proposition} \label{propo:Mink-Hl-Si-unif}
		On the set of all lattices of covolume $M$ in $\R^n$ there exists a probability measure $\nu$ such that, for any $w>0$
		\[
			\mathop{\mathbb E}_{\Lambda \sim \nu} \left(\Delta(u,\gunif{w}^\Lambda)\right) \leq \frac{1}{2}\;\sqrt{\frac{M}{\vol{w}}} .
		\]
		In particular, defining
		\[
			w_0 \eqdef \sqrt{n/2\pi e} \; M^{1/n},
		\]
		if $w > w_0$ we have 
		\[
			\mathop{\mathbb E}_{\Lambda \sim \nu} \left(\Delta(u,\gunif{w}^\Lambda)\right) \leq O(1)\;\left(\frac{w_0}{w}\right)^{n/2}.
		\]
	\end{proposition}

	\begin{proof}
		We define $\nu$ to be the procedure that samples a lattice according to $\mu$ of covolume $M^{-1}$, then outputs its dual. In the following chain, we first apply Proposition~\ref{propo:FBSDLat}; then, Jensen's inequality; then, the Minkowski-Hlawka-Siegel (MHS)~Theorem (Theorem~\ref{theo:Mink-Hl-Si}) to the function $|\widehat{\gunif{w}^{\Lambda}}|^2$; and, lastly, Parseval's Identity (Theorem~\ref{theo:Parseval}). This yields:
		\begin{align*}
			\mathop{\mathbb E}_{\Lambda \sim \nu} \left(2\Delta(u,\gunif{w}^\Lambda)\right) &\leq \mathop{\mathbb E}_{\Lambda^* \sim \mu} \left(\sqrt{\sum_{\vec x \in \Lambda^*\setminus\{\vec 0\}}|\widehat{\gunif{w}}(\vec x)|^2}\right) \quad \text{(Proposition~\ref{propo:FBSDLat})} \\ 
			&\leq \sqrt{\mathop{\mathbb E}_{\Lambda^* \sim \mu} \left(\sum_{\vec x \in \Lambda^*\setminus\{\vec 0\}}|\widehat{\gunif{w}}(\vec x)|^2\right)} \quad \text{(Jensen's Inequality)} \\
			&=\sqrt{\frac{1}{M^{-1}} \; \left(\int_{\R^n}|\widehat{\gunif{w}}(\vec x)|^2d\vec x\right)} \quad \text{(MHS Theorem)} \\
			&= \sqrt{M \int_{\R^n}|\gunif{w}(\vec x)|^2d\vec x} \quad \text{(Parseval's Identity)} \\
			&= \sqrt{\frac{M}{V_n(w)^2}\int_{\R^n} 1_{\Bc_{w}}(\vec x)d\vec x}\\
			& = \sqrt{\frac{M}{V_n(w)}} . 
		\end{align*}
		For the ``in particular'' part of the proposition, we use Stirling's estimate to derive
		\[
			\vol{w} = \frac{\pi^{n/2}\; w^n}{\Gamma(n/2+1)} = \frac{\pi^{n/2}\;w^n}{\left(\frac{n}{2e}\right)^{n/2}}\;(1+o(1))^n
		\]
		from which it follows that if 
		\[
			w > w_0 = \sqrt{n/2\pi e}\; M^{1/n} ,
		\]
		we have 
		\[
			\sqrt{\frac{M}{V_n(w)}} \leq O(1) \left(\frac{w}{w_0}\right)^{n/2}
		\]
		which concludes the proof. 
	\end{proof}

It is easily verified that the value of $w_{0}$ defined in Proposition~\ref{propo:Mink-Hl-Si-gauss} corresponds to the so-called Gaussian heuristic.
We view this condition on $w>w_{0}$ as the equivalent of the Gilbert-Varshamov bound for codes as we discussed just below Proposition~\ref{propo:BUnifRandCase}. In particular, as we need the support of the noise to have volume at least $M$ if we hope to smooth a lattice of covolume $M$, we see that the uniform distribution over a ball is optimal for smoothing random lattices, just as the uniform distribution over a sphere was optimal for smoothing random codes. 
	
	{\bf \noindent Gaussian Noise.} We now turn to the case of Gaussian noise. 
Following the proof of Proposition~\ref{propo:Mink-Hl-Si-unif} to the point where we apply Parseval's identity, but replacing $\gunif{w}$ by $D_s$, we obtain that 
	\[
		\mathbb E \left(\Delta(u,D_s^\Lambda)\right) \leq \sqrt{M \int_{\R^n}|D_s(\vec x)|^2d\vec x} \ .
	\]
	To conclude, one uses the following routine computation
	\begin{align*}
		\int_{\R^n}|D_s(\vec x)|^2d\vec x = \frac{1}{s^{2n}}\int_{\R^n}e^{-2\pi\left(\frac{|\vec x|_2}{s}\right)^2} d\vec x = 
\frac{1}{s^{2n}}\int_{\R^n}\rho_{s/\sqrt 2}(\vec x)d\vec x = \left(\frac{1}{s\sqrt{2}}\right)^n.
	\end{align*}
	Thus, we obtain:
	\begin{proposition} \label{propo:Mink-Hl-Si-gauss}
		On the set of all the lattices of covolume $M$ in $\R^n$ there exists a probability measure $\nu$ such that, for any $s>0$,
		\[
			\mathop{\mathbb E}_{\Lambda \sim \nu} \left(\Delta(u,D_s^\Lambda)\right) \leq \frac{1}{2}\; \sqrt{\frac{M}{\left(s\sqrt{2}\right)^n}} \ .
		\]
		In particular, if $s>s_0\eqdef M^{1/n}/\sqrt 2$, we have 
		\[
			\mathop{\mathbb E}_{\Lambda \sim \nu} \left(\Delta(u,D_s^\Lambda)\right) \leq \left(\frac{s_0}{s}\right)^{n/2}.
		\]
	\end{proposition}
	To compare Propositions \ref{propo:Mink-Hl-Si-unif} and \ref{propo:Mink-Hl-Si-gauss}, we note that a random vector sampled according to $D_{s}$ has an expected Euclidean norm given by $s\frac{\Gamma\left(\frac{n+1}{2}\right)}{\sqrt{\pi}\Gamma\left(\frac{n}{2}\right)} \sim s\sqrt{\frac{n}{2\pi}}$. So, it is fair to compare the effectiveness of smoothing with a parameter $s$ Gaussian distribution and the uniform distribution over a ball of radius $s\sqrt{\frac{n}{2\pi}}$. We note that, if $s_0$ is as in Proposition~\ref{propo:Mink-Hl-Si-gauss} and $w_0$ is the radius of the so-called Gaussian heuristic, then
	\[
		s_0\sqrt{\frac{n}{2\pi}} = \frac{M^{1/n}}{\sqrt 2} \sqrt{\frac{n}{2\pi}} = w_0 \; \sqrt{e/2} .
	\]
	Thus, we conclude that the parameter $s_0$ from Proposition~\ref{propo:Mink-Hl-Si-gauss} is larger than what we could hope 
by a factor 
$\sqrt{e/2}$.

	\subsection{Connecting Uniform Ball Distribution to Gaussian} 
	
	However, recall that in the code-case we argued that, as the Hamming weight of a vector sampled according to the Bernoulli distribution is tightly concentrated, we could obtain the same smoothing bound for the Bernoulli distribution as we did for the uniform sphere distribution, essentially by showing that we can approximate a Bernoulli distribution by a convex combination of uniform sphere distributions. Similarly, we can relate the Gaussian distribution to the uniform distribution over a ball, and thereby remove this additional $\sqrt{e/2}$ factor.

	We state a general proposition that allows us to translate smoothing bounds for the uniform ball distribution to the Gaussian distribution. It guarantees that if the uniform ball distribution smooths whenever $w>w_0$, the Gaussian distribution smooths whenever $s > w_0 \; \sqrt{\frac{2\pi}{n}}$. While the intuition for the argument is the same as that which we used in the code-case, the argument is itself a bit more sophisticated. 
	
	\begin{restatable}{proposition}{GaussianVSUnif} \label{propo:gauss-to-unif}
	Let $\Lambda$ be a random lattice of covolume $M$ and let $u \eqdef u_{\R^n/\Lambda}$ be the uniform distribution over its cosets. Suppose that for all $w>w_0$ there is a function $f(n)$ such that 
	\[
		\mathbb{E}_{\Lambda}\left(\Delta(u,\gunif{w}^{\Lambda})\right) \leq f(n) \left(\frac{w_0}{w}\right)^{n/2}.
	\]	
	Let $s_0 \eqdef w_0 \sqrt{\frac{2\pi}{n}}$. Then, for all $s>s_0$, defining $\eta \eqdef 1-\frac{s_0}{s} \in (0,1)$, we have
	\[
		\mathbb{E}_{\Lambda}\left(\Delta(u,D_s^{\Lambda})\right) \leq \exp(-\frac{\eta^2}{8} \; n) + f(n) \left(\frac{s_0}{s}\right)^{n/4}.
	\]
	\end{restatable}

	\begin{proof}
		See Appendix \ref{app:GaussianUnif}. 
	\end{proof}

	Combining the above proposition with Theorem~\ref{theo:Mink-Hl-Si}, setting $f(n)=O(1)$, we obtain the following theorem. 

	\begin{theorem} \label{theorem:random-lattice-gauss-to-unif}
	Let $\Lambda$ be a random lattice of covolume $M$ sampled according to $\nu$, let $u \eqdef u_{\R^n/\Lambda}$ be the uniform distribution over its cosets, and let \[s_0 \eqdef M^{1/n}/\sqrt{e}.\] 
	Then, for any $s>s_0$, setting $\eta \eqdef 1-\frac{s_0}{s} \in (0,1)$, we have 
	\[
		\mathbb{E}_{\Lambda}\left(\Delta(u,D_s^{\Lambda})\right) \leq \exp(-\frac{\eta^2}{8} \; n) + O(1) \left(\frac{s_0}{s}\right)^{n/4}.
	\]
	\end{theorem}
	
	\subsection{Smoothing Random $q$-ary Lattices} While the method of sampling lattices promised by the Minkowski-Hlawka-Siegel Theorem (Theorem~\ref{theo:Mink-Hl-Si}) is indeed very convenient for computations, it does not tell us much about how to explicitly sample from the distribution. Furthermore it is not very relevant if one is interested in the random lattices that are used in cryptography. 
	
	For a more concrete sampling procedure that is relevant to cryptography, we can consider the randomized Construction A (or, more precisely, its dual), which gives a very popular random model of lattices which are easily constructed from random codes. Specifically, for a prime $q$ and a linear code $\CC \subseteq (\Z/q\Z)^n$ we obtain a lattice as follows. First, we ``lift'' the codewords $\vec c \in \CC$ to vectors in $\R^n$ in the natural way by identifying $\Z/q\Z$ with the set $\{0,1,\dots,q-1\}$; denote the lifted vector as $\widetilde{\vec{c}}$. Then, we can define the following lattice
	\[
		\Lambda_{\CC} \eqdef \{\widetilde{\vec{c}} : \vec c \in \CC\} + q\Z^n.
	\]
	In other words: $\Lambda_{\CC}$ consists of all vectors in the integer lattice $\Z^n$ whose reductions modulo $q$ give an element of $\CC$. 
	
	Fix integers $1 \leq k \leq n$, a prime $q$ and a desired covolume $M$. We sample a random lattice $\Lambda$ as follows
	\begin{itemize}
		\item First, sample a random linear code $\CC\subseteq (\Z/q\Z)^n$ of dimension $k$ (recall this means that we sample a random $k \times n$ matrix $\vec G$ and define $\CC=\{\vec m \vec G:\vec m \in (\Z/q\Z)^k\}$),
		\item Then, we scale $\Lambda_{\CC}$ by $\frac{1}{M^{1/n}}\;\frac{1}{q^{1-k/n}}$,
		\item Lastly, we output the dual of $\frac{1}{M^{1/n}}\;\frac{1}{q^{1-k/n}}\Lambda_{\CC}$.
	\end{itemize}
	Notice that the scaling is chosen so that, as long as $\vec G$ is of full rank, the lattice $\Lambda$ we output has the desired covolume $M$. We denote this procedure of sampling $\Lambda$ by $\nu_{\textup{A}}$ (the dependence on $q$, $k$ and $n$ is left implicit).
	
	The important fact is that, up to an error term (which decreases as $q$ increases), the expected number of lattice points from $\dual \Lambda$ in a Euclidean ball of radius $r$ is roughly $\frac{\vol r}{M}$, as one would hope. 
	
	\begin{proposition}[{\cite[Lemma 7.9.2]{Z14}}] \label{propo:unif-of-q-ary-lattice}
		For every $n \geq 2$, $1 \leq k < n$ and prime power $q$, for $\Lambda \sim \nu_{\textup{A}}$ the expected number of lattice points from $\dual \Lambda$ in a Euclidean ball of radius $w \eqdef t\sqrt{n}$ satisfies
		\[
			\sqrt[n]{\frac{M\; \mathbb E_{\Lambda}(\Nb{\leq w}{\dual\Lambda})}{\vol{w}}} = 1 \pm \delta/t \quad \mbox{where } \delta \eqdef \frac{1}{q^{1-k/n}}. 
		\]
\end{proposition}

	We now turn to bounding the expected statistical distance between $u$ and $\gunif{w}^{\Lambda}$, where $\Lambda\sim\nu_{\textup{A}}$ and $w>0$ is the radius of the Euclidean ball from which the noise is uniformly sampled. First, we state an explicit formula for the Fourier transform of $1_{\Bc_{w}}$, the indicator function of a Euclidean ball of radius $w$, in terms of \emph{Bessel functions}. 
	
	\begin{notation}
		For a positive real number $\mu>0$, we denote by $J_\mu:\R\to\R$ the Bessel function of the first kind of order $\mu$.
	\end{notation} 
	
	The important fact concerning Bessel functions that we will use is the following.
	\begin{fact} We have
	\begin{align} \label{eq:fourier-transform-bessel}
		\widehat{1_{\Bc_w}}(\vec y) = \left(\frac{w}{|\vec y|_2}\right)^{n/2} J_{n/2}(2\pi w|\vec y|_2) .
	\end{align}
	\end{fact} 
	We will refrain from providing an explicit formula for Bessel functions, and instead use the following upper-bound as a black-box. 
	
	\begin{proposition}[\cite{K06}]\label{propo:bound-on-bessel}
		For any $x \in \R$ we have 
		\[
			|J_{n/2}(x)| \leq |x|^{-1/3} .
		\]
	\end{proposition}

	Using this proposition, we first prove a technical lemma that will be reused when we discuss smoothing arbitrary lattices. In order to state the lemma, we introduce the following auxiliary function.
	
	\begin{notation} \label{not:g_w}
		For a real $w>0$, we define $g_w:\R \to \R$ via 
		\[
			g_w(t) \eqdef \frac{1}{\vol w}\widehat{{1}_{\Bc_w}}(\vec x)^2
		\]
		where $\vec x$ is any vector in $\R^n$ of norm $t$. Note that as $\widehat{{1}_{\Bc_w}}(\vec x)$ depends only on $|\vec x|_2$, this is indeed well-defined. 
	\end{notation}
	
	The following lemma leverages Proposition~\ref{propo:bound-on-bessel} to upper-bound $g_w$ on a closed interval.  
	
	\begin{lemma} \label{lemma:varphi-bound}
		For any $w>0$ and any $0 \leq a$ and $b = \left(1+\frac{1}{n}\right)a$ we have, for some constant $C>0$
		\[
			\max_{a \leq t \leq b}g_w(t) \leq \frac{C}{\vol b w^{2/3}}\;\frac{1}{a^{2/3}} .
		\]
	\end{lemma}

	\begin{proof}
		 First, we notice that for all $t \in [a,b]$ 
		\begin{align*}
			\vol t = \left(\frac{t}{b}\right)^n\vol b \geq \left(\frac{a}{b}\right)^n\vol b = \left(1+\frac{1}{n}\right)^{-n}\vol b \geq \frac{1}{C'}\vol b
		\end{align*}
		for some constant $C'>0$. We now use Proposition~\ref{propo:bound-on-bessel} to derive
		\begin{align*}
			\max_{a \leq t \leq b}\;g_w(t) \leq \frac{C'}{\vol b} \;\max_{a \leq t \leq b} J_{n/2}(2\pi wt)^2 
\leq \frac{C}{\vol bw^{2/3}}\;\frac{1}{a^{2/3}}
		\end{align*}
	for an appropriate constant $C>0$ which concludes the proof.
	\end{proof}

	We now provide the main theorem of this section. It demonstrates that to smooth our ensemble of random $q$-ary codes (in expectation) with the uniform distribution over the ball of radius $w$, it still suffices to choose $w > w_0 \eqdef \sqrt{n2\pi/e}\;M^{1/n}$, assuming $q$ is not too small. 
	
	\begin{theorem} \label{theo:smoothing-q-ary}
		Let $n>2$ and $1 \leq k < n$. Let $q$ be a prime and set $\gamma \eqdef \frac{n^{3/2}}{q^{1-k/n}}$. Let $\Lambda \sim \nu_{\textup{A}}$. For some constant $C>0$, we have 
		\[
			\mathbb E_{\Lambda}\left(\Delta(u,\gunif{w}^{\Lambda})\right) \leq C \left(\frac{n}{w}\right)^{1/3} e^{\gamma/2} \sqrt{\frac{M}{\vol w}} .
		\]
		In particular, if $w>w_0 \eqdef \sqrt{n2\pi/e}M^{1/n}$, we have 
		\[
			\mathbb E_{\Lambda}\left(\Delta(u,\gunif{w}^{\Lambda})\right) \leq O\left(\left(\frac{n}{w}\right)^{1/3} e^{\gamma/2}\right) \left(\frac{w_0}{w}\right)^{n/2} .
		\]
	\end{theorem}
	
	\begin{proof}
		Let $t_j\eqdef\left(1+\frac{1}{n}\right)^j$ for $j\in\mathbb{N}$ and  
		\[
			N_j \eqdef \sharp\{\dual{\vec x} \in \dual \Lambda : t_j \leq |\dual{\vec x}|_2 < t_{j+1}\} \quad ; \quad \varphi_j \eqdef \max_{t_j \leq t \leq t_{j+1}}g_w(t).
		\]
Now, we apply Proposition~\ref{propo:FBSDLat} and the above definitions to obtain
		\begin{align*}
			\mathbb E_{\Lambda}\left(2\Delta(u,\gunif{w}^{\Lambda})\right) &\leq \mathbb E_{\Lambda}\left(\sqrt{\sum_{\vec x \in \dual{\Lambda}\setminus \{\vec 0\}}|\widehat{\gunif{w}}(\vec x)|^2}\right) \\
			&\leq \sqrt{\frac{1}{\vol w}\mathbb E_{\Lambda} \left(\sum_{\vec x \in \dual{\Lambda}\setminus \{\vec 0\}}g_w(\vec x)\right)} \quad (\mbox{Jensen's inequality}) \\
			&\leq \sqrt{\frac{1}{\vol w}\mathbb E_{\Lambda} \left(\sum_{j=0}^{\infty}N_j\varphi_j\right)} \\
			&\leq \sqrt{\frac{1}{\vol w} \sum_{j=0}^{\infty} \mathbb E\left(\Nb{\leq t_{j+1}}{\dual\Lambda}\right)\varphi_j} \ .
		\end{align*}
		By Proposition~\ref{propo:unif-of-q-ary-lattice}, we may upper-bound
		\begin{align}
			\mathbb E_{\Lambda}\left(\Nb{\leq t_{j+1}}{\dual\Lambda}\right) \leq M \; \vol{t_{j+1}}\left(1 + \left(\frac{\sqrt n}{\left(1+\frac{1}{n}\right)^{jn}q^{1-k/n}}\right)\right)^n \ .
		\end{align}
		Now, recalling $\gamma = \frac{n^{3/2}}{q^{1-k/n}}$ we have for any $j \geq 0$
		\[
			\left(1 + \left(\frac{\sqrt n}{\left(1+\frac{1}{n}\right)^{jn}q^{1-k/n}}\right)\right)^n \leq \left(1 + \left(\frac{\sqrt n}{q^{1-k/n}}\right)\right)^n \leq e^{n\frac{\sqrt{n}}{q^{1-k/n}}} = e^\gamma.
		\]
		Thus, we conclude 
		\begin{align*}
			\mathbb E_{\Lambda}\left(2\Delta(u,\gunif{w}^{\Lambda})\right) \leq \sqrt{\frac{e^\gamma M}{\vol w} \sum_{j=0}^{\infty}\vol{t_{j+1}})\varphi_j} \ .
		\end{align*}
		Now, by Lemma~\ref{lemma:varphi-bound} we have $\varphi_j \leq \frac{C_1}{\vol{t_{j+1}}w^{2/3}}\frac{1}{t_j^{2/3}}$ for all $j \geq 0$. Hence,
		\begin{align*}
			\sum_{j=0}^{\infty}\vol{t_{j+1}}\varphi_j &\leq \frac{C_1}{w^{2/3}}\sum_{j=0}^{\infty}\frac{\vol{t_{j+1}}}{\vol{t_{j+1}}}\frac{1}{t_j^{2/3}}\\
			&= \frac{C_1}{w^{2/3}}\sum_{j=0}^{\infty}\frac{1}{(1+1/n)^{2j/3}}\\
			&= \frac{C_1}{w^{2/3}}\frac{1}{1-(1+1/n)^{-2/3}}\\
			&\leq \frac{C_2 \; n^{2/3}}{w^{2/3}} \ ,
		\end{align*}
		for an appropriate constant $C_2>0$. Thus, putting everything together we derive
		\begin{align*}
			\mathbb E_{\Lambda}\left(\Delta(u,\gunif{w}^{\Lambda})\right) \leq \sqrt{\frac{e^\gamma M}{2\vol w}\; \frac{C_2 n^{2/3}}{w^{2/3}}}
			\leq C \left(\frac{n}{w}\right)^{1/3} e^{\gamma/2} \sqrt{\frac{M}{\vol w}} 
		\end{align*}
		for some constant $C>0$. The ``in particular'' part of the Theorem follows analogously to the corresponding argumentation (Stirling's estimate) used in the proof of Proposition \ref{propo:Mink-Hl-Si-unif}.
\end{proof}
	
	Next, turning to Gaussian noise, we could again prove a smoothing bound ``directly,'' but this will lose the same factor of $\sqrt{e/2}$ as we had earlier. Instead, we apply Proposition~\ref{propo:gauss-to-unif} with the function $f(n) = O\left(\left(\frac{n}{w}\right)^{1/3} e^{\gamma/2}\right)$ to conclude the following.
	
	\begin{theorem} \label{theorem:random-q-ary-gauss-to-unif}
		Let $n>2$ and $1 \leq k < n$. Let $q$ be a prime and set $\gamma \eqdef \frac{n^{3/2}}{q^{1-k/n}}$. Let $\Lambda$ be a random $q$-ary lattice sampled according to $\nu_A$, let $u=u_{\R^n/\Lambda}$ be the uniform distribution over its cosets, and let 
		$$
		s_0 \eqdef M^{1/n}/\sqrt{e}.
		$$ 
		Then, for any $s>s_0$, setting $\eta \eqdef 1-\frac{s_0}{s} \in (0,1)$, we have 
		\[ 						
			\mathbb{E}_{\Lambda}\left(\Delta\left(u,D_s^\Lambda\right)\right) \leq \exp\left(-\frac{\eta^2}{8} \; n\right) + O(1)\;(s/s_0)^{n/4} \; e^{\gamma/2}. 
		\]
	\end{theorem}

	\subsection{Smoothing Arbitrary Lattices}
	
	We now turn our attention to the task of smoothing arbitrary lattices.

	Analogously to how we used the minimum distance of the dual code to give our smoothing bound for worst-case codes, we will use the shortest vector of the dual lattice in order to provide our smoothing bound for worst-case lattices. The lemma that we will apply is the following where
	$$
	\CKL \eqdef 2^{0.401}.
	$$ 
	\begin{lemma}[{\cite[Lemma 3]{PS09}}]\label{lemma:KLt}
		For any $n$-dimensional lattice $\Lambda$,
		$$
		\forall t \geq \lambda_{1}(\Lambda), \quad N_{\leq t}(\Lambda) \leq \frac{\vol{t}}{\vol{\lambda_{1}(\Lambda)}} \; \CKL^{n(1+o(1))}.
		$$
	\end{lemma}

	\begin{remark}
		This lemma is a consequence of the Kabatiansky and Levenshtein' bound \cite{KL78} on the size of spherical codes, historically known as the ``second linear programming bound''. It is why we may refer to the aforementioned bound of Lemma \ref{lemma:KLt} as the second linear programming bound.  
	\end{remark}

We begin by considering the effectiveness of smoothing with noise uniformly sampled from the ball.
The following theorem is proved using similar techniques to those we used for Theorem~\ref{theo:smoothing-q-ary}, although instead of using Proposition~\ref{propo:unif-of-q-ary-lattice} to bound the $\Nb{\leq t}{\dual \Lambda}$'s, we use Lemma~\ref{lemma:KLt}.
	
	\begin{theorem}\label{theo:bSDEuc} 	Let $\Lambda$ be an $n$-dimensional lattice and $u \eqdef u_{\mathbb R^n/\Lambda}$ be the uniform distribution over its cosets. Then, it holds that
		$$
	\Delta\left(u,\gunif{w}^{\Lambda} \right) \leq  
\sqrt{\frac{\CKL^{n(1+o(1))}}{\vol{\lambda_{1}(\dual{\Lambda})} \; \vol{w}} }.  
		$$
		In particular, setting
		\[w_0 \eqdef n \; \frac {\CKL^{1+o(1/n)}} {2 \pi \; e \; \lambda_1(\Lambda^*)}\]
		for all $w>w_0$, it holds that 
	\[	
		\Delta\left(u,\gunif{w}^{\Lambda} \right) \leq O(1) (w_0/w)^{n/2}.
	\]
	\end{theorem}

	\begin{proof}
		Define 
		$$
			t_0 \eqdef \lambda_1(\dual \Lambda), \quad t_{j+1} \eqdef \left(1+\tfrac{1}{n}\right)t_j \quad \mbox{and} \quad \varphi_j\eqdef \max_{t_j \leq t \leq t_{j+1}}\{g_w(t)\} ~~\text{ for } j\geq 0,
		$$
		where we recall the definition of $g_w(t) = \frac{1}{\vol w}\widehat{{1}_{\Bc_w}}(\vec x)^2$ with $|\vec x|_2=t$ (see Notation~\ref{not:g_w}). We also define
		\[
			N_j \eqdef \sharp\{\dual{\vec x} \in \dual \Lambda: t_j \leq |\dual{\vec x}|_2 \leq t_{j+1}\} \ .
		\]
		With this notation and Proposition~\ref{propo:FBSDLat} we have 
		\begin{align}
			2\Delta\left(u,\gunif{w}^{\Lambda} \right) &\leq \sqrt{\sum_{\vec x \in \dual{\Lambda}\setminus \{\vec 0\}}|\widehat{\gunif{w}}(\vec x)|^2} \nonumber \\
			&\leq \sqrt{\frac{1}{\vol w}\sum_{\vec x \in \dual{\Lambda}\setminus \{\vec 0\}}g_w(\vec x)} \nonumber \\
			&\leq \sqrt{\frac{1}{\vol w}\sum_{j=0}^{\infty}N_j\varphi_j} \nonumber \\
			&\leq \sqrt{\frac{1}{\vol w}\sum_{j=0}^\infty \Nb{\leq t_{j+1}}{\dual \Lambda} \varphi_j} \label{eq:worst-case-smooth-start} \ .
		\end{align}  
		By Lemma~\ref{lemma:varphi-bound}, for some constant $C_1>0$, we obtain
		\[
			\varphi_j \leq \frac{C_1}{V_n(t_{j+1})w^{2/3}}\frac{1}{t_j^{2/3}} \ .
		\]
		Combining this with the upper-bound on $\Nb{\leq t_{j+1}}{\dual \Lambda}$ provided by Lemma~\ref{lemma:KLt} (note that $t_{j+1} \geq \lambda_1(\dual \Lambda)$ for all $j \geq 0$), we find 
		\begin{align*}
			\sum_{j=0}^\infty \Nb{\leq t_{j+1}}{\dual \Lambda} \varphi_j &\leq \sum_{j=0}^\infty \frac{V_n(t_{j+1})}{V_n(\lambda_1(\dual \Lambda))}\CKL^{n(1+o(1))}\; \frac{C_1}{V_n(t_{j+1})w^{2/3}}\; \frac{1}{t_j^{2/3}} \\
			&= \frac{\CKL^{n(1+o(1))}}{V_n(\lambda_1(\dual \Lambda))w^{2/3}}\;\sum_{j=0}^{\infty}\frac{1}{t_{j}^{2/3}} \\
			&= \frac{\CKL^{n(1+o(1))}}{V_n(\lambda_1(\dual \Lambda))w^{2/3}}\; \sum_{j=0}^{\infty}\frac{1}{\lambda_1(\dual \Lambda)^{2/3}\left(1+\frac{1}{n}\right)^{2j/3}} \\
			&\leq \frac{\CKL^{n(1+o(1))}}{V_n(\lambda_1(\dual \Lambda))w^{2/3}}\; \left(\frac{n}{w \lambda_{1}(\dual{\Lambda})}\right)^{2/3}.
 		\end{align*}
		In the above, all necessary constants were absorbed into the $\CKL^{o(n)}$ term. Combining this with \eqref{eq:worst-case-smooth-start}, we obtain the first part of the theorem. The ``in particular'' part again follows using Stirling's approximation. 
	\end{proof}

	Next, we can consider the effectiveness of smoothing with the Gaussian distribution. As usual, we could follow the steps of the proof of Theorem~\ref{theo:bSDEuc} and obtain the same result, but with an additional multiplicative factor of $\sqrt{\frac{e}{2}}$. That is, we obtain

	\begin{theorem}\label{theo:bSDEgauss}
	Let $\Lambda$ be an $n$-dimensional lattice and $u \eqdef u_{\mathbb R^n/\Lambda}$ be the uniform distribution over its cosets. Then, it holds.
		$$
		\Delta\left(u, D_{s}^{\Lambda} \right) \leq  \sqrt{\frac{\CKL^{n(1+o(1))}}{\vol{\lambda_{1}(\dual{\Lambda})} \; \vol{s\sqrt{n/(2\pi)}} } \; \left(\frac{e}{2}\right)^{n/2}}.
		$$

	In particular, setting \[s_0 \eqdef \sqrt n \; \frac {\CKL^{1+o(1/n)}} {2\sqrt {\pi e} \; \lambda_1(\Lambda^*)}, \]	
 	it holds for any $s> s_0$ that $\Delta\left(u, D_{s}^{\Lambda} \right) \leq O(1) \; (s_0 / s)^{n/2}$.
	\end{theorem}

	However, as usual it is more effective to combine the bound for the uniform ball distribution and decompose the Gaussian as a convex combination of uniform ball distributions, {\em i.e.} to apply Proposition~\ref{propo:gauss-to-unif}. In this way, we can obtain the following theorem, improving the smoothing bound $s_0$ by another $\sqrt {e/2}$ factor. In the following theorem, we are setting the $f(n)$ function of Proposition~\ref{propo:gauss-to-unif} with the $O(1)$ term in the bound of Theorem~\ref{theo:bSDEuc}.

\begin{theorem} \label{theo:bSDEgauss_better}
	Let $\Lambda$ be an $n$-dimensional lattice, $u \eqdef  u_{\mathbb R^n/\Lambda}$ the uniform distribution over its cosets, and
	\[
		s_0 \eqdef \sqrt n \; \frac{\CKL^{1+o(1/n)}}{\sqrt{2\pi} \; e \; \lambda_1(\Lambda^*)} \ . 
	\]
	Then, for any $s>s_0$ and letting $\eta\eqdef= 1 - \frac{s_0}{s} \in (0, 1)$, it holds that 
	\[\Delta\left(u, D_{s}^{\Lambda} \right) \leq \exp\left(- \frac{\eta^{2}}{8} \; n\right) + O(1) \; \left(\frac {s_0} {s}\right)^{n/4}.\] 
\end{theorem}
 	
	\section{Acknowledgement}
	We would like to thank Iosif Pinelis for help with the proof of Proposition~\ref{propo:gauss-to-unif}. 
	
\bibliographystyle{alpha}

	\appendix
	
	\section{Proof of Proposition \ref{propo:BerVSUnif}}\label{app:BerVSUnif}

Our aim in this section is to prove the following proposition 

\BerVSUnif*

	Roughly speaking, this proposition is a consequence of the fact that a Bernoulli distribution concentrates Hamming weights over a small number of slices close to the expected weight (here $np$) and, on each slice the Bernoulli distribution is uniform. Let us introduce the truncated Bernoulli distribution over words of Hamming weight $[(1-\varepsilon)pn,(1+\varepsilon)pn]$ for some $\varepsilon > 0$, namely
 \begin{equation}\label{eq:fTrunc}
 \fberTrunc{p}(\vec{x}) \eqdef 	\left\{
\begin{array}{ll}
	\frac{1}{Z} \; \fber{p}(\vec{x}) & \mbox{if } |\vec{x}| \in \left[ (1-\varepsilon)pn,(1+\varepsilon)pn \right]  \\
	0 & \mbox{otherwise.}
\end{array}
\right. 
\end{equation}
where 
\begin{equation}
\label{eq:Z}
Z \eqdef \mathop{\sum}\limits_{|\vec{y}| = (1-\varepsilon)np}^{(1+\varepsilon)np} \fber{p}(\vec{y})
\end{equation} is the probability normalizing constant.

	Proposition \ref{propo:BerVSUnif} is a consequence of the following lemmas.

	\begin{lemma}\label{lemma:gepsFber}
		Let $\varepsilon>0$. We have
		$$
		\Delta\left(\fber{p},	\fberTrunc{p} \right) = 2^{-\Omega(n)} . 
		$$
	\end{lemma}

	\begin{proof}By Chernoff's bound
		\begin{equation}\label{eq:chernoff} 
			1-Z =\sum_{\substack{\vec{y} : \\ |\vec{y}| \notin \left[ (1-\varepsilon)np,(1+\varepsilon)np \right]}} \fber{p}(\vec{y}) \leq 2e^{-\varepsilon^{2}n} = 2^{-\Omega(n)} . 
		\end{equation}
	Therefore for any $|\vec{x}| \in \left[ (1-\varepsilon)np,(1+\varepsilon)np \right]$,
	\begin{align}\label{eq:geps}
		\fberTrunc{p}(\vec{x}) &= \frac{1}{1-2^{-\Omega(n)}} \; \fber{p}(\vec{x}) \nonumber \\
		&= \left(1+2^{-\Omega(n)} \right) \; \fber{p}(\vec{x}) . 
	\end{align}
	We have now the following computation:
	\begin{align*}
		2\Delta\left(\fber{p},\fberTrunc{p}  \right) &= \sum_{\vec{x}\in\F_{2}^{n}} \left| \fber{p}(\vec{x}) - \fberTrunc{p} (\vec{x}) \right| \\ 
		&=  \sum_{|\vec{x}| \in \left[ (1-\varepsilon)np,(1+\varepsilon)np \right]} \left| \fber{p}(\vec{x}) - \fberTrunc{p} (\vec{x}) \right|   + \sum_{|\vec{x}| \notin \left[ (1-\varepsilon)np,(1+\varepsilon)np \right]} \left| \fber{p}(\vec{x}) \right|   \\
		&= 2^{-\Omega(n)}\left(  \sum_{|\vec{x}| \in \left[ (1-\varepsilon)np,(1+\varepsilon)np \right]} \left| \fber{p}(\vec{x}) \right| \right)  + 2^{-\Omega(n)} \quad \mbox{(Equations \eqref{eq:chernoff} and \eqref{eq:geps})}  \\
		&= 2^{-\Omega(n)}
	\end{align*}
	where in the last line we used that $\fber{p}$ is a probability distribution.	
	\end{proof}

	\begin{lemma}\label{lemma:gepsCodeFber}We have
		$$
		\Delta\left(u,\fber{p}^{\CC}\right) \leq \Delta\left(u,\fberTrunc{p} ^{\CC}\right) + 2^{-\Omega(n)}.
		$$
	\end{lemma}
	\begin{proof}
		By the triangle inequality, 
		\[
			\Delta\left(u,\fber{p}^{\CC}\right) \leq \Delta\left(u,\fberTrunc{p}^{\CC}\right) + \Delta\left(\fber{p}^{\CC},\fberTrunc{p}^{\CC}\right) . 
		\]
		Focusing on the second term now
		\begin{align*}
			\Delta\left(\fber{p}^{\CC},\fberTrunc{p}^{\CC}\right) &= \frac{1}{2} \sum_{\vec{y} \in \F_2^n/\CC}\left|\fber{p}^{\CC}(\vec{y}) - \fberTrunc{p}^{\CC}(\vec{y})\right| \\
&= \frac{1}{2} \sum_{\vec{y} \in \F_2^n/\CC}\left|\sum_{\vec c \in \CC}\fber{p}(\vec{c}+\vec y) - \sum_{\vec c \in \CC}\fberTrunc{p}(\vec{c}+\vec y)\right| \\
&\leq \frac{1}{2} \sum_{\vec{y} \in \F_2^n/\CC}\sum_{\vec c \in \CC}\left|\fber{p}(\vec c + \vec y) - \fberTrunc{p}(\vec c + \vec y) \right| \\
&= \Delta\left(\fber{p},\fberTrunc{p}\right).
		\end{align*}
	which concludes the proof by Lemma \ref{lemma:gepsFber}. 
	\end{proof}

	The following lemma is a basic property of the statistical distance. 

	\begin{lemma}\label{lemma:statIneqCvx} For any distribution  $f$ and $(g_i)_{1 \leq i \leq m}$ we have
		$$
		\Delta\left(f, \sum_{i=1}^{m}\lambda_{i} g_{i} \right) \leq \sum_{i=1}^{m} \lambda_{i} \; \Delta(f,g_{i})
		$$
		where the $\lambda_{i}$'s are positive and sum to one. 
	\end{lemma}

	We are now ready to prove Proposition \ref{propo:BerVSUnif}.

	\begin{proof}[Proof of Proposition \ref{propo:BerVSUnif}] First, by Lemma \ref{lemma:gepsCodeFber} we have
		\begin{equation}\label{eq:uFber}
			\Delta\left(u, \fber{p}^{\CC}\right) \leq \Delta\left(u, \fberTrunc{p}^{\CC}\right) + 2^{-\Omega(n)}.
		\end{equation}

	To upper-bound $\Delta\left(u, \fberTrunc{p}^{\CC}\right)$ we are going to use Lemma \ref{lemma:statIneqCvx}. 
	Notice that
	$$
	\fber{p} = \sum_{r=0}^{n} \binom{n}{r}p^{r}(1-p)^{n-r} \unifs{r}.
	$$
	Therefore it is readily seen that 
	\begin{equation*} 
	\fberTrunc{p} = \sum_{r = (1-\varepsilon)np}^{(1+\varepsilon)np } \lambda_{r} \; \unifs{r} \quad \mbox{where} \quad \lambda_{r} \eqdef \frac{1}{ Z} \; \binom{n}{r}p^{r}(1-p)^{n-r}. 
	\end{equation*} 
	By using Lemma~\ref{lemma:statIneqCvx} we obtain:
	\begin{align}
		\Delta\left(u, \fberTrunc{p}^{\CC}\right)  &\leq \sum_{r = (1-\varepsilon)np}^{(1+\varepsilon)np} \lambda_{r} \; \Delta\left(u,\unifs{r}^{\CC}\right) \nonumber\\
		&\leq \sum_{r = (1-\varepsilon)np}^{(1+\varepsilon)np} \Delta\left(u, \unifs{r}^{\CC}\right)\label{eq:gCodeu}
	\end{align}
	where in the last line we used that the $\lambda_{r}$'s are smaller than one. To conclude the proof we plug Equation \eqref{eq:gCodeu} in \eqref{eq:uFber}. 
	\end{proof}

 	\section{Proof of Proposition \ref{propo:ABL}}\label{app:proofPropoABL}

Our aim in this section is to prove the following proposition which is an extension of \cite[Theorem 3]{ABL01} for $\tau \in [\delta,1]$ (\cite[Theorem 3]{ABL01} only applied for $\tau \in [\delta,1/2]$.)

\PropoABL*

Our proof is mainly a rewriting of the proof of  \cite[Theorem 3]{ABL01} which relies on the following proposition.

	\begin{proposition}[{\cite[Proposition $2$ with $d' = 0$]{ABL01}}]\label{propo:BoundBarg}
	Let $\CC$ be a binary code of length $n$ such that $\dmin(\CC) = \Omega(n)$. Let $t \eqdef  \frac{n}{2} - \sqrt{\dmin(\CC)(n-\dmin(\CC))}$ and $a$ be such that
	$$
	x_{1}^{(t+1)} < a < x_{1}^{(t)} \quad \mbox{;} \quad \frac{K_{t}(a)}{K_{t+1}(a)} = -1
	$$
	where $x_{1}^{(\mu)}$ denotes the first root of the Krawtchouk polynomial of order $\mu$, namely $K_{\mu}$. 
	
	When $0 \leq w < t \leq n/2$, we have
	\begin{equation}
		\sum_{\vec{c} \in \CC \backslash \{\mathbf{0}\}} K_{w}(|{\vec{c}}|)^{2} \leq \frac{t+1}{2a} \; \frac{\binom{n}{w}}{\binom{n}{t}} \left( \binom{n}{t+1} + \binom{n}{t} \right)^{2}
	\end{equation}
\end{proposition}

The approach is to optimize on the choice of $w$ in Proposition \ref{propo:BoundBarg} to give an 
upper-bound on $N_{\ell}(\CC)$. More precisely we  observe that
\begin{equation}\label{eq:toBoundL1} 
	\Neq{\ell}{\CC} \leq \frac{1}{K_{w}(\ell)^{2}} \sum_{\vec{c}\in\CC \backslash \{\mathbf{0}\}} K_{w}(|\vec{c}|)^{2} \leq \frac{1}{K_{w}(\ell)^{2}}\; \frac{t+1}{2}\frac{\binom{n}{w}}{\binom{n}{t}} \left( \binom{n}{t+1} + \binom{n}{t} \right)^{2}
\end{equation}
and then choose $w$ to minimize $\frac{\binom{n}{w}}{K_{w}(\ell)^{2}}$.

	\begin{proof}[Proof of Proposition \ref{propo:ABL}]
	It will be helpful to bring in  the following map:
		$$
		x\in[0,1] \mapsto x^{\perp} \eqdef \frac{1}{2} - \sqrt{x(1-x)}.
		$$
		It can be verified that this application is  an involution, is symmetric $(1-x)^\perp = x^\perp$ and decreasing on $[0,\frac{1}{2}]$.

		Let $\CC$ be a binary code of length $n$ such that $\dmin(\CC) = \delta n$ where $\delta\in(0,1/2]$ and $t$ be defined as in Proposition \ref{propo:BoundBarg}. 	
	Let 
$\omega \eqdef \frac{w}{n}, \lambda \eqdef\frac{\ell}{n}$ and  $\delta^{\perp} \eqdef 1/2 - \sqrt{\delta(1-\delta)}$.
	Then by Proposition \ref{propo:BoundBarg} we have (see Equation \eqref{eq:toBoundL1})
	\begin{equation}\label{eq:KrawUPB} 
	\frac{\log_2 \Neq{\ell}{\CC} }{n}  \leq  h(\omega) + h(\delta^{\perp}) - \frac{2 \log_{2} |K_{w}(\ell)|}{n} + o(1).
\end{equation}
\noindent
{\bf Case 1: $\lambda \in [\delta,1-\delta]$.}	 \\
It is optimal to choose in this case $w$ such that $\omega= \lambda^\perp - \varepsilon$ where
$\varepsilon > 0$ and $\varepsilon = o(1)$ as $n$ tends to infinity.  Let us first notice that  $ \lambda \in [\delta,1-\delta]$ implies that $\lambda^\perp \leq \delta^\perp$ which together with $\omega < \lambda^\perp$ implies that $\omega < \delta^\perp$ which in turn is equivalent to the condition $w <t$ for being able to apply Proposition \ref{propo:BoundBarg}. Moreover $\omega < \lambda^\perp$ also implies $\lambda < \omega^\perp$ and by using Proposition \ref{prop:expansion} we obtain
$$
\frac{2\log_{2} |K_{w}(\ell)|}{n} \leq   h(\omega) + 1 - h(\lambda) +o(1).$$
Therefore
$$\frac{\log_2 \Neq{\ell}{\CC} }{n}  \leq  h(\omega) + h(\delta^{\perp})  - h(\omega) -1  + h(\lambda) +o(1)=
 h(\delta^{\perp})   + h(\lambda) -1 +o(1).
 $$
 
 {\bf Case 2: $\lambda \in (1-\delta,1]$.}\\
 In that case, let $\omega= \delta^{\perp} - \varepsilon$ with $\varepsilon > 0$ and $\varepsilon = o(1)$ as $n$ tends to infinity.
 Here we can write
 $$\frac{2\log_{2} |K_{w}(\ell)|}{n} = \frac{\log_{2} (K_{w}(\ell)^2)}{n}= \frac{\log_{2} (K_{w}(n-\ell)^2)}{n}.$$
 Since $\lambda > 1 - \delta$, we have $1-\lambda < \delta$. On the other hand, $\omega< \delta^{\perp}$ implies $\delta < \omega^\perp$. We deduce from these two inequalities that $1 - \lambda < \omega^\perp$. By using Proposition \ref{prop:expansion} again, we get
 $$
 \frac{\log_{2} (K_{w}(n-\ell)^2)}{n} = 2a(1-\lambda,\delta^\perp)+o(1)=2a(\lambda,\delta^\perp)+o(1).
 $$
 By plugging this estimate in \eqref{eq:KrawUPB} we get
 $$
 \frac{\log_2 \Neq{\ell}{\CC} }{n} \leq  2 h(\delta^{\perp}) - 2 a(\lambda,\delta^{\perp}).
 $$
This concludes the proof. 
\end{proof}

 	\section{Proof of Theorem \ref{theo:finalUBSD} }\label{app:proofThFinalCode}

Our aim in this appendix is to prove the following theorem. 

\thFinalCode*

{\bf Sketch of proof.}
We will use the following proof strategy
	\begin{itemize}
        \item[1.] By Lemma \ref{lemma:gepsCodeFber} we know that on one hand
	\begin{equation}
          \label{eq:difference}
		\Delta\left(u,\fber{p}^{\CC}\right) = \Delta\left(u,\fberTrunc{p}^{\CC}\right) + 2^{-\Omega(n)}.
              \end{equation}
This is actually a consequence of Chernoff's bound. This argument can also be used to show that the Fourier transforms are also close to each other pointwise
	\begin{equation}\label{eq:TFBer}
	\forall \vec{x}\in \F_{2}^{n}, \quad 2^{n}\;  \left|\widehat{\fberTrunc{p}}(\vec{x}) - \widehat{\fber{p}}(\vec{x})\right| = 2^{-\Omega(n)}.
	\end{equation} 

	\item[2.] Equation \eqref{eq:TFBer} together with Lemma \ref{lemma:LCodewords} are then used to show that:
	\begin{equation}\label{eq:step1} 
	\Delta\left(u,\fberTrunc{p}^{\CC}\right) \leq 2^n\sqrt{\sum_{t = \dmin(\dual{\CC})}^{n-\dmin(\dual{\CC})/2} \Neq{t}{\dual{\CC}} \widehat{,\fberTrunc{p}}(t)^{2}} + 2^{-\Omega(n)}.
	\end{equation} 	
	\item[3.] We use the two previous points to upper-bound $\Delta\left(u,\fber{p}^{\CC}\right)$ as in the  equation above 
          and conclude by using bounds of Propositions \ref{propo:ABL} and \ref{propo:2LPB}.
\end{itemize}

\noindent
{\bf Proof of Step 1.} 
As we explained above \eqref{eq:difference} is just Lemma \ref{lemma:gepsCodeFber}.
	Let us now prove that

\begin{lemma}\label{lemma:lemmfBerVSTrunc}
	We have
	$$
		\forall \vec{x}\in \F_{2}^{n}, \quad 2^{n}\;  \left|\widehat{\fberTrunc{p}}(\vec{x}) - \widehat{\fber{p}}(\vec{x})\right| = 2^{-\Omega(n)}.
	$$
\end{lemma} 

	\begin{proof} 
		Recall that $Z = \mathop{\sum}\limits_{|\vec{y}| = (1-\varepsilon)np}^{(1+\varepsilon)np} \fber{p}(\vec{y})$ where by Chernoff's bound, we have 
		\begin{equation}\label{eq:M}
			Z = 1 - 2^{-\Omega(n)}.
		\end{equation}
	Notice now that,
	$$
	\fber{p} = \sum_{r=0}^{n} \binom{n}{r}p^{r}(1-p)^{n-r} \unifs{r} \quad \mbox{and} \quad \fberTrunc{p} = \frac{1}{Z}\;\sum_{r=(1-\varepsilon)pn}^{(1+\varepsilon)pn} \binom{n}{r}p^{r}(1-p)^{n-r} \unifs{r}/
	$$
	Let $\mathcal{I} \eqdef \llbracket (1-\varepsilon)pn, (1+\varepsilon)pn \rrbracket$. Notice that $Z = \sum_{r\in \mathcal{I}} \binom{n}{r}p^{r}(1-p)^{n-r}$. By linearity of the Fourier transform we obtain the following computation:
	\begin{align}
		\left|\widehat{\fberTrunc{p}}(\vec{x}) - \widehat{\fber{p}}(\vec{x})\right| &= \left( \frac{1}{Z} - 1 \right) \sum_{r\in \mathcal{I}} \binom{n}{r}p^{r}(1-p)^{n-r} \left| \widehat{\unifs{r}}(\vec{x}) \right| \nonumber \\
		& \qquad\qquad\qquad\qquad+ \sum_{r\notin \mathcal{I}} \binom{n}{r}p^{r}(1-p)^{n-r} \left| \widehat{\unifs{r}}(\vec{x}) \right| \nonumber\\
		&= 2^{-\Omega(n)} \sum_{r \in \mathcal{I}}\binom{n}{r}p^{r}(1-p)^{n-r} \left| \widehat{\unifs{r}}(\vec{x}) \right|
+ 2^{-\Omega(n)} \max_{r} \left| \widehat{\unifs{r}}(\vec{x}) \right|
\label{ineq:truncBer}
	\end{align}
where in the last line we used Equation \eqref{eq:M}. 
	Recall now that by definition of the Fourier transform for functions over $\F_{2}^{n}$ we have:
	$$
	\left| \unifs{r}(\vec{x}) \right| = \left| \frac{1}{2^{n}} \sum_{\vec{y} : |\vec{y}|=r} \frac{(-1)^{\vec{x}\cdot\vec{y}}}{\binom{n}{r}} \right| \leq \frac{1}{2^{n}}.
	$$
	By plugging this in Equation \eqref{ineq:truncBer} we get:
		\begin{align*}
		\left|\widehat{\fberTrunc{p}}(\vec{x}) - \widehat{\fber{p}}(\vec{x})\right| &\leq  \frac{2^{-\Omega(n)}}{2^{n}} \underbrace{\sum_{r\in \mathcal{I}} \binom{n}{r}p^{r}(1-p)^{n-r}}_{\leq 1} + \frac{2^{-\Omega(n)}}{2^{n}} \\
&= \frac{2^{-\Omega(n)}}{2^{n}}
		\end{align*}
	which concludes the proof. 
	\end{proof} 

{\bf Proof of Step 2.}
This corresponds to proving the following lemma.
\begin{lemma}\label{lem:step1}
\begin{equation*}
		\Delta\left(u,\fberTrunc{p}^{\CC}\right) \leq 2^n\sqrt{\sum_{t = \dmin(\dual{\CC})}^{n-\dmin(\dual{\CC})/2} \Neq{t}{\dual{\CC}} \widehat{,\fberTrunc{p}}(t)^{2}} + 2^{-\Omega(n)}.
\end{equation*}
\end{lemma}

	\begin{proof}
By applying Proposition \ref{propo:FBSDCod} to $\fberTrunc{p}$ we obtain
	\begin{equation}\label{eq:boundfBerTrunc1}
	\Delta\left(u,\fberTrunc{p} ^{\CC}\right) \leq 2^{n} \sqrt{\sum_{t = \dmin(\dual{\CC})}^{n} \Neq{t}{\dual{\CC}}|\widehat{\fberTrunc{p}}(t)|^{2}}
	\end{equation}
	where $\widehat{\fberTrunc{p}}(t)$ denotes the common value of the radial function $\widehat{\fberTrunc{p}}$ on vectors of Hamming weight $t$.
	Recall now that $\widehat{\fber{p}}(\vec{x}) = \frac{1}{2^{n}}\; (1-2p)^{|\vec{x}|}$ and by Lemma \ref{lemma:lemmfBerVSTrunc} that $2^{n}\;  \left|\widehat{\fberTrunc{p}}(\vec{x}) - \widehat{\fber{p}}(\vec{x})\right| = 2^{-\Omega(n)}$. Therefore, 
	$$
	\forall \vec{x}\in\F_{2}^{n}, \mbox{ } |\vec{x}|\geq n-\frac{\dmin(\dual{\CC})}{2} \quad \mbox{:} \quad  2^{n}\; \left|\widehat{\fberTrunc{p}}(\vec{x})\right| = 2^{-\Omega(n)}.
	$$
	By plugging this in Equation \eqref{eq:boundfBerTrunc1} we obtain (as there is at most one dual codeword of weight $\ell$ for each $\ell >n-\dmin(\dual{\CC})/2$, see Lemma \ref{lemma:LCodewords})
		\begin{equation}\label{eq:boundfBerTrunc2}
		\Delta\left(u,\fberTrunc{p} ^{\CC}\right) \leq 2^{n} \sqrt{\sum_{t = \dmin(\dual{\CC})}^{n - \dmin(\dual{\CC})/2} \Neq{t}{\dual{\CC}}|\widehat{\fberTrunc{p}}(t)|^{2}} + 2^{-\Omega(n)}
	\end{equation}
which concludes the proof. 
\end{proof}

{\bf Proof of Step 3.}
We finish the proof of Theorem \ref{theo:finalUBSD} by noticing that
	$$
	\fberTrunc{p} = \frac{1}{Z}\;\sum_{\ell=(1-\varepsilon)pn}^{(1+\varepsilon)pn} \binom{n}{\ell}p^{\ell}(1-p)^{n-\ell} \unifs{\ell}
	$$
	where $Z \eqdef  \mathop{\sum}\limits_{|\vec{y}| = (1-\varepsilon)np}^{(1+\varepsilon)np} \fber{p}(\vec{y}) = 1 - 2^{-\Omega(n)}$ by Chernoff's bound. Therefore, 
	$$
	\widehat{\fberTrunc{p}} = \left( 1+ 2^{-\Omega(n)}\right)\;\sum_{\ell=(1-\varepsilon)pn}^{(1+\varepsilon)pn} \binom{n}{\ell}p^{\ell}(1-p)^{n-\ell} \;\widehat{\unifs{\ell}}.
	$$
	By plugging this in Equation \eqref{eq:boundfBerTrunc2} and using  $\widehat{\unifs{\ell}} = \frac{1}{2^{n}} \; \frac{K_{\ell}}{\binom{n}{\ell}}$ we obtain
	$$
		\Delta\left(u,\fberTrunc{p}^{\CC}\right) \leq  \left( 1+2^{-\Omega(n)}\right) \; \sqrt{\sum_{t = \dmin(\dual{\CC})}^{n-\dmin(\dual{\CC})/2} \Neq{t}{\dual{\CC}}\left( \sum_{\ell = (1-\varepsilon)pn}^{(1+\varepsilon)pn} p^{\ell}(1-p)^{n-\ell} K_{\ell}(t) \right)^{2}} + 2^{-\Omega(n)}.
	$$

We then use in the righthand term, Propositions \ref{propo:ABL}, \ref{propo:2LPB} which give bounds on the $\frac{1}{n} \; \log_{2} N_{\ell}(\dual{\CC})$'s (where $\dmin(\dual{\CC}) \geq \dual{\delta}n$) and Proposition \ref{prop:expansion} which gives an asymptotic expansion of Krawtchouk polynomials to upper-bound $\Delta\left(u,\fberTrunc{p}^{\CC}\right)$. We finish the proof of the theorem by using 
	this upper-bound in the righthand term of \eqref{eq:difference}.

 	\section{Proof of Proposition~\ref{propo:gauss-to-unif}}\label{app:GaussianUnif}

Our aim in this section is to prove the following proposition.

\GaussianVSUnif*

It will be a consequence of the following lemmas. We begin with the following result decomposing the Gaussian as a convex combination of balls. 
\begin{lemma}
	\label{lemma:gaussian_convex_combi_ball}
	The Gaussian distribution in dimension $n$ of parameter $s$ is the following convex combination of uniform distributions over balls:
	\[ D_s = \frac {1}{s} \int_{0}^\infty G_n(w/s) \; \gunif{w}\, dw \]
	where $G_{n}(x) = x^{n+1} \; \vol{1} \; 2\pi\; \exp\left(-\pi x^2\right) \geq 0$. Furthermore, we have $\frac{1}{s}\int_{0}^{\infty}G_n(w/s) \, dw = 1$.
\end{lemma}

\begin{proof}
	First, let $g_s(w) \eqdef \frac{1}{s^n}\; \exp\left(-\pi\tfrac{w^2}{s^2}\right)$ ({\em i.e.} the value the probability density function $D_s$ takes on vectors of weight $w$) and denote $h_s(w) = -g_s'(w) = \frac{2\pi w}{s^{n+2}} \; \exp\left(-\pi\tfrac{w^2}{s^2}\right)$. For any $\vec x \in \R^n$, setting $u = |\vec x|_2$, as $\lim_{w \to \infty} g_s(w)=0$ we have
	\begin{align*}
		D_s(\vec x) &= g_s(u) = \int_{u}^\infty h_s(w) \, dw = \int_{0}^{\infty}h_s(w)\; 1\{u \leq w\} \ dw = \int_{0}^{\infty}h_s(w) \; 1_{\mathcal{B}_{w}}(\vec x)\ dw \ .
	\end{align*}
	Above, we denoted by $1\{u\leq w\}$ the function which takes value $1$ on input $w$ if $u \leq w$, and $0$ otherwise. To conclude, note that $\frac{1}{s}\; G_n(w/s) = h_s(w) \; \vol{w}$ and recall $\gunif{w} = \frac{1_{\mathcal B_w}}{\vol{w}}$. 
	
	For the ``furthermore'' part of the lemma, we compute  
	\begin{align} \label{eq:int-pre-sub}
		\frac{1}{s}\int_{0}^{\infty}G_n(w/s) \, dw = \frac{1}{s}\int_{0}^{\infty} (w/s)^{n+1} \; \vol{1} \; 2\pi\; \exp\left(-\pi (w/s)^2\right)\, dw \ .
	\end{align}
	We make the substitution $t = \pi \left(\frac{w}{s}\right)^2$, which means $dw = \frac{s^2\, dt}{2\pi w} = \frac{s}{2\sqrt{t\pi}}\, dt$. Also, we recall $\vol{1} = \frac{\pi^{n/2}}{\Gamma(n/2+1)}$. Thus, 
	\begin{align*}
		\frac{1}{s}\int_{0}^{\infty}G_n(w/s) \, dw &= \frac{1}{s}\; \frac{\pi^{n/2}}{\Gamma(n/2+1)}\int_{0}^{\infty} \left(\frac{t}{\pi}\right)^{(n+1)/2} \; 2\pi \; e^{-t} \; \frac{s}{2\sqrt{t\pi}}\, dt \\
		&= \frac{1}{\Gamma(n/2+1)}\int_{0}^{\infty}t^{n/2} \; e^{-t} \; dt = \frac{\Gamma(n/2+1)}{\Gamma(n/2+1)} = 1
\end{align*}
	which concludes the proof. 
\end{proof}

	We now quote the following bound, which makes precise the intuition that it is exponentially unlikely that a random Gaussian vector has norm $(1-\eta)$ factor smaller than its expected norm. This result provides the analogy for the Chernoff bound that we used for the code-case. 

\begin{lemma} [{\cite[Example 2.5]{W19}}]\label{propo:gaussian-tail-bound}
	Let $\vec X$ be a random Gaussian vector of dimension $n$ and parameter $1$. Let $0 < \eta < 1$. Then
	\[
	\mathbb{P}\left(|\vec X|_{2}^2 \leq (1-\eta)\;\frac{n}{2\pi}\right) \leq \exp(-\frac{\eta^2}{8} \; n).
	\]
\end{lemma}

This lemma allows us to prove the following lemma bounding $\frac{1}{s}\int_{0}^{\overline w}G_n(w/s)dw$ when $\overline w < s \; \sqrt{n/(2\pi)}$. 

\begin{lemma} \label{lem:bound_G_n}
	Let $\eta \in (0,1)$ and $\overline w = \sqrt{1-\eta} \; s \; \sqrt{n/(2\pi)}$. Then
	\[
	\frac{1}{s}\int_{0}^{\overline w}G_n(w/s)dw \leq 
\exp(-\frac{\eta^2}{8} \; n) \ .
	\]
\end{lemma}

\begin{proof}
	Let $\overline u \eqdef  \sqrt{1-\eta} \; \sqrt{n/(2\pi)}$. By Lemma~\ref{propo:gaussian-tail-bound}, if $\vec X$ denotes a random Gaussian vector of dimension $n$ and parameter $1$, we have 
	\begin{equation} \label{ineq:ini}
		\int_{0 \leq |\vec x|_{2} \leq \overline u}\exp(-\pi |\vec x|_{2}^2) \ d\vec x =  \mathbb{P}\left(|\vec X|_{2}^2 \leq (1-\eta)\;\frac{n}{2\pi}\right) \leq \exp(-\frac{\eta^2}{8} \; n).
	\end{equation} 
	To compute this last integral, note that 
	\begin{align}\label{eq:1} 
		\int_{0 \leq |\vec x|_{2} \leq \overline u}\exp(-\pi |\vec x|_{2}^2) \ d\vec x &= \int_0^{\overline u} \int_{u\mathcal{S}^{n-1}} e^{-\pi u^2}dA du \ ,
	\end{align}
	where $u\mathcal{S}^{n-1}$ denotes the Euclidean sphere of radius $u$ and $dA$ is the area element. If $A_{n-1}(u)$ denotes the surface area of $u\mathcal{S}^{n-1}$, then $A_{n-1}(u) = u^{n-1} A_{n-1}(1)$ and thus
	\begin{align}\label{eq:2} 
		\int_0^{\overline u} \int_{u\mathcal{S}^{n-1}} e^{-\pi u^2}dA du = A_{n-1}(1)\int_{0}^{\overline u}u^{n-1}\exp(-\pi u^2) \, du \ .
	\end{align}
	Further, it is known that $A_{n-1}(1) = \frac{2\pi^{n/2}}{\Gamma(n/2)}$. Therefore, plugging Equations \eqref{eq:1} and \eqref{eq:2} into \eqref{ineq:ini} leads to
	\begin{equation}\label{ineq:R} 
	\int_{0}^{\overline u}u^{n-1}\exp(-\pi u^2) \, du \leq \frac{1}{A_{n-1}(1)} \; \exp(-\frac{\eta^2}{8} \; n) \ .
	\end{equation}
	
	Now, we look at the left-hand side of the inequality we wish to prove. We begin by making the substitution $u = w/s$. So then $dw = s\;du$. Moreover, let $\overline u \eqdef \sqrt{1-\eta} \; \sqrt{n/(2\pi)}$ and note that when $w = \overline w$ we have $u = \overline w/s = \sqrt{1-\eta} \; \sqrt{n/(2\pi)} = \overline u$. 
	\begin{align*}
		\frac{1}{s}\int_{0}^{\overline w} G_n(w/s) \ dw &= \int_{0}^{\overline u}G_n(u) \ du \\
		&= \vol{1} \; 2\pi \int_{0}^{\bar u}  u^{n+1} \; \exp(-\pi u^2) \ du \\
		&\leq \vol{1} \; 2\pi\bar{u}^2 \int_{0}^{\bar u} u^{n-1} \exp(-\pi u^2) \ du\ .
	\end{align*}
	Plugging this last inequality with \eqref{ineq:R} yields
	\begin{align*}
		\frac{1}{s}\int_{0}^{\overline w} G_n(w/s) \ dw \leq \frac{\vol{1} 2\pi \; \overline{u}^2}{A_{n-1}(1)} \exp(-\frac{\eta^2}{8} \; n) \ .
	\end{align*}
	To conclude the proof, note that $V_n(1)=\int_0^1 A_{n-1}(u) du= \int_0^1 u^{n-1} A_{n-1}(1) du=\frac{A_{n-1}(1)}{n}$ and therefore 
	\[
	\frac{\vol{1} 2\pi \; \overline{u}^2}{A_{n-1}(1)} = \frac{2 \pi (1-\eta)n}{2 \pi n}=1-\eta \leq 1.
\]
It concludes the proof. 
\end{proof}

We are now ready to prove Proposition \ref{propo:gauss-to-unif}.

\begin{proof}[Proof of Proposition \ref{propo:gauss-to-unif}.]
	
		By Lemma~\ref{lemma:gaussian_convex_combi_ball}, $D_s$ is a convex combination of uniform distribution over balls, namely
$D_s = \frac{1}{s}\int_0^{\infty}G_n(w/s) \; \gunif{w} \;dw$. 
Therefore (we use here the analogue of Lemma \ref{lemma:statIneqCvx} in the context of the statistical distance between two probability density functions)
		\[
		\mathbb{E}_{\Lambda}\left(\Delta(u,D_s^{\Lambda})\right) \leq \frac{1}{s}\int_0^{\infty} G_n(w/s)\; \mathbb{E}_{\Lambda}\left(\Delta(u,\gunif{w}^\Lambda)\right) dw.
		\]
		We split the integral in two parts at radius $\overline w = \sqrt{1-\eta} \; s \; \sqrt{n/(2\pi)}$. For the first part $w \leq \overline w$, we use the trivial bound $\mathbb{E}_{\Lambda}\left(\Delta(u,\gunif{w}^\Lambda)\right) \leq 1$ which gives:
		\[
		\frac{1}{s}\int_0^{\overline w} G_n(w/s)\; \mathbb{E}_{\Lambda}\left(\Delta(u,\gunif{w}^\Lambda)\right) dw \leq \frac{1}{s} \int_0^{\overline w} G_n(w/s)dw .
		\]
		We then apply Lemma~\ref{lem:bound_G_n}, which bounds this part by $\exp(-\frac{\eta^2}{8} \; n)$. 
		
		For the second part $w \geq \overline w$, we use the trivial bound $\frac 1s \int_{\overline w}^{\infty}G_n(w/s)dw \leq 1$ and, noting 
		\[
		w \geq \overline w = \sqrt{1-\eta} \; s \; \sqrt{n/(2\pi)} = \frac{1}{\sqrt{1-\eta}} \; s_0 \; \sqrt{n/(2\pi)} > s_0 \; \sqrt{n/(2\pi)} = w_0,
		\]
		we may apply the assumption of the proposition, yielding
		\begin{align*}
			\mathbb{E}_{\Lambda}\left(\Delta(u,\gunif{w}^{\Lambda})\right) &\leq f(n)\left(\frac{w_0}{w}\right)^{n/2} \leq f(n) \left(\frac{w_0}{\overline w}\right)^{1/2} = f(n) \left(\sqrt{1-\eta}\right)^{n/2} = f(n) \left(\frac{s_0}{s}\right)^{n/4}.
		\end{align*}	
		Adding these bounds yields the proposition. 
\end{proof}

\end{document}